\pdfoutput=1
\documentclass[sigconf, nonacm=true, authorversion=true]{acmart}

\usepackage{balance}

\copyrightyear{2023}
\acmYear{2023}
\setcopyright{rightsretained}
\acmConference[CCS '23]{Proceedings of the 2023 ACM SIGSAC Conference on
Computer and Communications Security}{November 26--30, 2023}{Copenhagen,
Denmark}
\acmBooktitle{Proceedings of the 2023 ACM SIGSAC Conference on Computer
and Communications Security (CCS '23), November 26--30, 2023, Copenhagen,
Denmark}
\acmDOI{10.1145/3576915.3623192}
\acmISBN{979-8-4007-0050-7/23/11}

\newcommand{\shortVersion}{false}

\usepackage{multicol}
\usepackage{amsmath}

\usepackage{pgf}
\usepackage{tikz}
\usetikzlibrary{arrows,automata}
\usetikzlibrary{positioning}
\usetikzlibrary{shapes.geometric}
\usetikzlibrary{plotmarks}
\usetikzlibrary{patterns}
\usetikzlibrary{calc}
\usetikzlibrary{fit}
\usetikzlibrary{shapes.misc}
\usetikzlibrary{backgrounds}

\usetikzlibrary{pgfplots.groupplots}
\usetikzlibrary{shapes.multipart}

\usepackage{amsthm}

\makeatletter
\newtheoremstyle{sig}
  {}
  {}
  {\itshape}
  {}
  {\scshape}
  {.}
  {.5em}
  {#1\@ifnotempty{#2}{ #2}\thmnote{\quad(#3)}}%
\makeatother

\theoremstyle{sig}

\usepackage{mathpartir}
\usepackage{thmtools}
\usepackage{thm-restate}
\usepackage{paralist}
\usepackage{listings}
\usepackage{color}
\usepackage{graphicx}
\usepackage{float}
\usepackage{textcomp}
\usepackage{xspace}
\usepackage{pifont}
\usepackage{multirow}
\usepackage{algorithm}
\usepackage[noend]{algpseudocode}
\usepackage{pgfplots,booktabs}
\usepackage{enumitem}
\usetikzlibrary{snakes}
\tikzstyle{printersafe}=[snake=snake,segment amplitude=0 pt]
\usepackage{dirtree}
\usepackage{listings}
\usepackage{syntax}
\usepackage{bm}
\usepackage{stmaryrd}
\usepackage{mathtools}
\usepackage{diagbox}
\usepackage{tabularx}
\usepackage{xifthen}
\usepackage{cleveref}
\usepackage{color}
\usepackage{xcolor}
\usepackage{threeparttable}
\usepackage{bbding}

\usepackage{fancyhdr}
\usepackage[normalem]{ulem}

\usepackage{enumitem}
\setlistdepth{15} 
\newlist{longitemize}{itemize}{15}
\setlist[longitemize,1]{label=\textbullet}
\setlist[longitemize,2]{label=\textbullet}
\setlist[longitemize,3]{label=\textbullet}
\setlist[longitemize,4]{label=\textbullet}
\setlist[longitemize,5]{label=\textbullet}
\setlist[longitemize,6]{label=\textbullet}
\setlist[longitemize,7]{label=\textbullet}
\setlist[longitemize,8]{label=\textbullet}
\setlist[longitemize,9]{label=\textbullet}
\setlist[longitemize,10]{label=\textbullet}
\setlist[longitemize,11]{label=\textbullet}
\setlist[longitemize,12]{label=\textbullet}
\setlist[longitemize,13]{label=\textbullet}
\setlist[longitemize,14]{label=\textbullet}
\setlist[longitemize,15]{label=\textbullet}

\usepackage{rotating}
\usepackage{multirow}

\usepackage{subcaption}

\newcommand{\techReportAppendix}[1]{%
	\ifthenelse{\equal{\shortVersion}{true}}%
		{\cite{techReport}}%
        {Appendix~\ref{#1}}%
}%

\newcommand{\techReportAppendices}[2]{%
	\ifthenelse{\equal{\shortVersion}{true}}%
		{\cite{techReport}}%
		{Appendices~\ref{#1}--\ref{#2}}%
}%

\newcommand{\onlyTechReport}[1]{%
	\ifthenelse{\equal{\shortVersion}{true}}%
		{}%
		{#1}%
}%

\newcommand{\onlyShortVersion}[1]{%
	\ifthenelse{\equal{\shortVersion}{true}}%
		{#1}%
		{}%
}%

\makeatletter
\newcommand*\@lbracket{[}
\newcommand*\@rbracket{]}
\newcommand*\@colon{:}
\newcommand*\colorIndex{%
    \edef\@temp{\the\lst@token}%
    \ifx\@temp\@lbracket \color{black}%
    \else\ifx\@temp\@rbracket \color{black}%
    \else\ifx\@temp\@colon \color{black}%
    \else \color{orange}%
    \fi\fi\fi
}
\makeatother

\lstdefinestyle{verilog-style}
{
    language=Verilog,
    basicstyle=\footnotesize\ttfamily,
    keywordstyle=\color{blue},
    identifierstyle=\color{black},
    commentstyle=\color{cadmiumgreen},
    numbers=left,
	morekeywords = [2]{leak, on, monitor},
	keywordstyle = [2]\color{punct}\ttfamily,
    numberstyle=\tiny\color{black},
    extendedchars=true,
    numbersep=10pt,
    tabsize=4,
	frame=none,
    moredelim=*[s][\colorIndex]{[}{]},
    literate=*{:}{:}1,
    xleftmargin=5.0ex,
    captionpos=b,
    escapechar=\$,
    escapeinside={(*}{*)}
}
\colorlet{punct}{red!60!black}
\definecolor{delim}{RGB}{20,105,176}
\colorlet{numb}{magenta!60!black}
\definecolor{lightgreen}{HTML}{669900}		%
\definecolor{bluegreen}{HTML}{33997e}		%
\definecolor{brightube}{rgb}{0.82, 0.62, 0.91}
\definecolor{aquamarine}{rgb}{0.5, 1.0, 0.83}
\definecolor{cadmiumgreen}{rgb}{0.0, 0.42, 0.24}

\newcommand{\mypara}[1]{\smallskip\noindent\emph{\textbf{{#1.}}}}
\newcommand{\var}[1]{\mbox{\texttt{#1}}\xspace}
\newcommand{\theSet}[1]{\{#1\}}

\newcommand{\ie}{\emph{i.e.},\xspace}
\newcommand{\eg}{\emph{e.g.},\xspace}
\providecommand{\eqdef}{\triangleq{}}

\definecolor{ao(english)}{rgb}{0.0, 0.5, 0.0}
\definecolor{royalblue(web)}{rgb}{0.25, 0.41, 0.88}

\newcommand{\setCommentColor}[1]{%
	\ifthenelse{\equal{#1}{bk}}%
		{\colorlet{colorVar}{red!50}}%
		{\ifthenelse{\equal{#1}{kg}}%
			{\colorlet{colorVar}{blue}}%
			{\ifthenelse{\equal{#1}{mg}}%
				{\colorlet{colorVar}{ao(english)}}%
			{\ifthenelse{\equal{#1}{jr}}%
				{\colorlet{colorVar}{magenta}}%
				{\ifthenelse{\equal{#1}{zw}}%
					{\colorlet{colorVar}{orange}}%
					{\ifthenelse{\equal{#1}{gm}}%
						{\colorlet{colorVar}{cyan}}
						{}
					}%
				}%
			}%
		}%
	}%
}

\newcommand{\commentAuthor}[1]{%
	\ifthenelse{\equal{#1}{bk}}%
		{Boris:\ }%
		{\ifthenelse{\equal{#1}{kg}}%
			{Klaus:\ }%
			{\ifthenelse{\equal{#1}{mg}}%
				{Marco:\ }%
			{\ifthenelse{\equal{#1}{jr}}%
				{Jan:\ }%
				{\ifthenelse{\equal{#1}{zw}}%
					{Zilong:\ }%
					{\ifthenelse{\equal{#1}{gm}}%
						{Gideon:\ }%
						{}
					}%
				}%
			}%
		}%
	}%
}

\definecolor{codegray}{gray}{0.95}

\theoremstyle{definition}
\newtheorem{definition}{Definition}
\newtheorem{example}{Example}
\newtheorem{contract}{Contract}
\theoremstyle{theorem}
\newtheorem{theorem}{Theorem}

\newcommand{\emptysequence}{\varepsilon}
\newcommand{\concat}{\cdot}

\renewcommand{\paragraph}[1]{\smallskip\noindent\textbf{#1:\ }}

\newcommand{\arch}{\mathsf{arch}}

\newcommand{\Var}{\mathit{Regs}}
\newcommand{\Val}{\mathit{Vals}}
\newcommand{\Nat}{\mathbb{N}}

\newcommand{\kywd}[1]{\mathbf{#1}}

\definecolor{Blue3}{HTML}{0000CD}
\newcommand{\obsKywd}[1]{\textcolor{Blue3}{\mathtt{#1}}}

\newcommand{\startObsKywd}[1]{\obsKywd{start}}
\newcommand{\commitObsKywd}[1]{\obsKywd{commit}} 
\newcommand{\rollbackObsKywd}[1]{\obsKywd{rollback}}

\newcommand{\commitObs}[1]{\commitObsKywd{}} %
\newcommand{\rollbackObs}[1]{\rollbackObsKywd{}} %
\newcommand{\unaryOp}[1]{\ominus #1}
\newcommand{\binaryOp}[2]{#1 \otimes #2}

\newcommand{\passign}[2]{#1 \leftarrow #2}

\newcommand{\ite}[3]{\mathbf{if} \; #1 \; \mathbf{th} \; #2 \; \mathbf{el} \; #3}

\newcommand{\fetch}[1]{
\ifthenelse{\equal{#1}{}}{\kywd{fetch}}{\kywd{fetch}}%
}

\definecolor{Blue3}{HTML}{0000CD}
\definecolor{Green4}{HTML}{008B00}
\definecolor{Red3}{HTML}{CD0000}
\definecolor{orange}{rgb}{0.8, 0.47, 0.196}
\lstdefinestyle{Cstyle}
{
	frame = tb,
  belowskip=.4\baselineskip,
  aboveskip=.4\baselineskip,
  	showstringspaces = false,
  	breaklines = true,
  	breakatwhitespace = true,
  	tabsize = 3,
  	numbers = left,
    stepnumber = 1,
    numberstyle = \tiny\color{gray},
    language = {[ANSI]C},
    alsoletter={.\$},
    basicstyle={\ttfamily\color{black}},
    keywordstyle={\ttfamily\color{Blue3}},
    keywordstyle=[2]{\ttfamily\color{Green4}},
    keywordstyle=[3]{\ttfamily\color{orange}},
    keywordstyle=[4]{\ttfamily\color{violet}},
    otherkeywords = {skip,not},
    morekeywords = [2]{A,B},
    morekeywords = [3]{},
	morekeywords = [4]{y,x,z,w, size,size_A,k,temp},
	morecomment=[l][\small\itshape\color{purple!40!black}]{//},
	sensitive=true,
}

\usepackage{letltxmacro}
\newcommand*{\SavedLstInline}{}
\LetLtxMacro\SavedLstInline\lstinline
\DeclareRobustCommand*{\lstinline}{%
  \ifmmode
    \let\SavedBGroup\bgroup
    \def\bgroup{%
      \let\bgroup\SavedBGroup
      \hbox\bgroup
    }%
  \fi
  \SavedLstInline
}

\newcommand{\SimpISA}{\textsl{\textsc{sISA}}}
\newcommand{\SimpIMP}{\textsl{\textsc{sImpl}}}
\newcommand{\SimpLM}{\textsl{\textsc{sLM}}}
\newcommand{\SimpATK}{\textsl{\textsc{sAtk}}}

\newtheorem*{theorem*}{Theorem}

\newcommand{\tool}{\textsc{LeaVe}}

\newcommand{\product}[2]{#1 \times #2}

\newcommand{\future}{\circ}

\newcommand{\alwaysFuture}{\square}
\newcommand{\boundedFuture}[1]{{\alwaysFuture}^{#1}}

\newcommand{\lang}{\mu\textsc{Vlog}}
\newcommand{\val}{\mu}
\newcommand{\states}[1]{\mathit{states}(#1)}
\newcommand{\hwEval}[2]{\llbracket #1\rrbracket(#2)}
\newcommand{\hwEvalC}[3]{\llbracket #1\rrbracket_{#3}(#2)}
\newcommand{\hwEvalInfty}[2]{\llbracket #1\rrbracket^{\infty}(#2)}
\newcommand{\hwEvalInftyFilter}[3]{\llbracket #1\rrbracket^{\infty}|#3(#2)}

\newcommand{\readVars}[1]{\mathit{read}(#1)}
\newcommand{\writeVars}[1]{\mathit{write}(#1)}
\newcommand{\wireVars}[1]{\mathit{wires}(#1)}
\newcommand{\vars}[1]{\mathit{vars}(#1)}
\newcommand{\compose}[2]{#1 \triangleright #2}
\renewcommand{\compose}[2]{#1[#2]}		%
\newcommand{\archVars}{\textsl{\textsc{Arch}}}
\renewcommand{\arch}{\textsl{\textsc{ISA}}}
\newcommand{\uarch}{\textsl{\textsc{Impl}}}
\newcommand{\uarchVars}{\mu\archVars}

\newcommand{\isasat}[3]{#1 \equiv_{#3} #2}
\renewcommand{\isasat}[3]{#1 \vdash_{#3} #2}
\newcommand{\project}[2]{#1\mathord{\upharpoonright}_{#2}}
\newcommand{\atk}{\textsc{Atk}}
\renewcommand{\contract}{\textsl{\textsc{LM}}}
\newcommand{\leakorder}[4]{#1 \succeq_{#3}^{#4} #2}
\newcommand{\equivVars}[1]{\sim_{#1}}
\newcommand{\initStates}[1]{\mathit{init}(#1)}

\newcommand{\assignments}[1]{{#1}.A}
\newcommand{\outputs}[1]{{#1}.O}

\crefalias{condition}{equation}
\crefname{condition}{condition}{conditions}
\Crefname{condition}{Condition}{Conditions}
\creflabelformat{condition}{#2\textup{(#1)}#3}

\crefalias{appendix}{section}
\crefname{appendix}{appendix}{appendices}
\Crefname{appendix}{Appendix}{Appendices}
\creflabelformat{appendix}{#2\textup{#1}#3}

\algrenewcommand\algorithmicrequire{\textbf{Input:}}
\algrenewcommand\algorithmicensure{\textbf{Output:}}

\newcommand{\ctrsat}[2]{#1 \sqsupseteq #2} %

\newcommand{\runA}[1]{{#1}^1}
\newcommand{\runB}[1]{{#1}^2}

\newcommand{\RelInvs}{\mathit{RI}}
\newcommand{\CandInvs}{\mathit{CI}}

\newcommand{\stutteringProduct}[2]{#1 \times_{#2} #1}

\crefformat{section}{\S#2#1#3}
\crefmultiformat{section}{\S\S#2#1#3}{ and~#2#1#3}{, #2#1#3}{, and~#2#1#3}
\crefformat{subsection}{\S#2#1#3}
\crefmultiformat{subsection}{\S\S#2#1#3}{ and~#2#1#3}{, #2#1#3}{, and~#2#1#3}

\title{Specification and Verification of Side-channel Security for Open-source Processors via Leakage Contracts}

\author{Zilong Wang}
\affiliation{%
  \institution{IMDEA Software Institute\\Universidad Polit\'ecnica de Madrid}
  \city{Madrid}
  \country{Spain}
}
\email{zilong.wang@imdea.org}

\author{Gideon Mohr}
\affiliation{%
  \institution{Saarland University}
  \city{Saarbr\"ucken}
  \country{Germany}
}
\email{s8gimohr@stud.uni-saarland.de}

\author{Klaus von Gleissenthall}
\affiliation{%
  \institution{Vrije Universiteit Amsterdam}
  \city{Amsterdam}
  \country{Netherlands}
}
\email{k.freiherrvongleissenthal@vu.nl}
  
\author{Jan Reineke}
\affiliation{%
\institution{Saarland University}
\city{Saarbr\"ucken}
\country{Germany}
}
\email{reineke@cs.uni-saarland.de}

\author{Marco Guarnieri}
\affiliation{%
\institution{IMDEA Software Institute}
\city{Madrid}
\country{Spain}
}
\email{marco.guarnieri@imdea.org}


\begin{document}

\begin{abstract}
    Leakage contracts have recently been proposed as a new security abstraction at the Instruction Set Architecture (ISA) level.
    Leakage contracts aim to capture the information that processors leak through their microarchitectural implementations.
    However, so far, we lack a methodology to verify that a processor actually satisfies a given leakage contract.
    
    In this paper, we address this challenge by developing \tool{}, the first tool for verifying register-transfer-level (RTL) processor designs against ISA-level leakage contracts.
    To this end, we show how to decouple security and functional correctness concerns. %
    \tool{} leverages this decoupling to make verification of contract satisfaction practical.
    To scale to realistic processor designs, \tool{} further employs inductive reasoning on relational abstractions. %
    Using \tool{}, we precisely characterize the side-channel security guarantees of three open-source RISC-V processors, thereby obtaining the first proofs of contract satisfaction for RTL processor designs.\looseness=-1
\end{abstract}


\maketitle

\section{Introduction}\label{sec:intro}

Microarchitectural attacks~\cite{spectre2019,Lipp2018,RIDL, Bulck2018, Yarom14}
compromise security by exploiting software-visible artifacts of
microarchitectural optimizations like caches and speculative execution. 
To use modern hardware securely, programmers must be aware of how these optimizations impact the security of their code.
Unfortunately, instruction set architectures (ISAs), the traditional abstraction layer between hardware and software, do not provide an adequate basis for secure programming: ISAs capture the functional behavior of processors but abstract away microarchitectural details and thus fail to capture their security implications.

To build secure software systems on top of modern hardware, we need a new abstraction at the ISA level that faithfully captures the information processors may leak through their microarchitectural implementations.
We refer to this new abstraction as \emph{leakage contracts}.
For example, the leakage contract underlying {constant-time programming}~\cite{almeida2016verifying}, used for writing cryptographic code, states that processors can leak a program's control flow and memory accesses, which therefore must not depend on secret data.

Recent work has made significant strides towards using leakage contracts as a basis for building secure systems, through their formal specification~\cite{contracts2021,mosier2022axiomatic}; through automatic security analysis of software~\cite{spectector2020,fabian202automatic,pitchfork, GuancialeBD20, blade}; and through post-silicon processor fuzzing~\cite{oleksenko2022revizor,oleksenko2023hide,buiras2021micro,Nemati2020a}.
However, leakage contracts can only unfold their full potential once hardware is available that \emph{provably satisfies such contracts}. %
The proliferation of open-source processors around the RISC-V ecosystem presents an opportunity to fill this gap. %

In this paper, we present the first approach for verifying register-transfer-level (RTL) processor designs against ISA-level leakage contracts.
This requires overcoming the following challenges:\looseness=-1
\begin{asparaitem}%
	\item Bridging the abstraction gap between sequential instruction-level leakage contracts and cycle-level processor designs that overlap the execution of multiple instructions.
	
    \item Leakage contracts capture a processor's information leakage on top of its functional specification. 
    	Verifying contract satisfaction, thus, requires reasoning about both functional and security aspects, which goes against the separation of these two concerns. %

	\item Even simple open-source processor designs have large and complex state spaces, which prohibit explicit enumeration or bounded model checking. %
\end{asparaitem}

Our verification approach and its implementation \tool{} overcome these challenges based on the following contributions:
\begin{asparaenum}%
\item We adapt the leakage contract framework from~\cite{contracts2021} to RTL processor designs, capturing instruction-level contracts and realistic cycle-level attacker models in a single uniform framework.
\item We introduce a decoupling theorem that separates security and functional correctness aspects for contract satisfaction.
\item We develop a verification algorithm for checking the security aspects of contract satisfaction that employs inductive reasoning on relational abstractions to scale to realistic processor designs.\looseness=-1
\item We implement and experimentally evaluate our approach on three open-source RISC-V processors. %
\end{asparaenum}
Next, we discuss these four contributions in more detail. 

\paragraph{Leakage contracts for RTL processors}
We adapt the leakage contract framework from~\cite{contracts2021} for RTL processor designs (\Cref{sec:model}).
This requires significant changes since the framework in \cite{contracts2021} builds on top of a simple sequential operational model of an out-of-order processor rather than on cycle-level RTL circuits.
In a nutshell, we model both the instruction-level leakage contract and the microarchitectural attacker as \emph{monitoring circuits}. %
These monitoring circuits generate {contract traces}, capturing the processor's \emph{intended leakage} at instruction level, and {attacker traces}, capturing its \emph{actual leakage} at microarchitectural level.
In this setting, a microarchitecture \emph{satisfies a contract} for a given attacker if the following holds: whenever two architectural states yield different attacker traces, then the two states also yield different contract traces.

\paragraph{Decoupling security and functional correctness}
We introduce a decoupling theorem~(\Cref{sec:verification:decoupling}) that separates security and functional correctness concerns for contract satisfaction.
For this, we introduce the notion of \emph{microarchitectural contract satisfaction} that refers only to the microarchitecture and ensures the absence of leaks.
The decoupling theorem states that, for processors correctly implementing the instruction set architecture, contract satisfaction and microarchitectural contract satisfaction are equivalent.
This allows us to focus \emph{only} on the security challenges arising from leakage verification, while relying on existing approaches for functional correctness~\cite{reid2016end, Huang19, Zeng21,burch1994automatic, khune2010automated,patankar1999formal,jhala12001microarchitecture}.

\paragraph{Verifying contract satisfaction}
We develop a novel algorithm for checking microarchitectural contract satisfaction (\Cref{sec:verification:algorithm}), which we prove sound.
That is,  whenever our algorithm concludes that a contract is satisfied, then microarchitectural contract satisfaction indeed holds.
Given a contract monitoring circuit and a microarchitecture, our approach inductively learns invariants associated with pairs of microarchitectural executions with the same contract traces  using invariant learning techniques~\cite{flanagan2001houdini} and uses these invariants to establish contract satisfaction.

\paragraph{Implementation and evaluation}
We implement our approach in \tool{}, a tool for verifying microarchitectural contract satisfaction for processor designs in Verilog (\Cref{sec:implementation}).
We validate our approach by precisely characterizing the side-channel security guarantees of three open-source RISC-V processors in multiple configurations (\Cref{sec:evaluation}).
For this, we define a family of leakage contracts capturing leaks through control flow, memory accesses, and variable-time instructions, and use \tool{} to determine which contracts each processor satisfies against an attacker observing when instructions retire.
Our evaluation confirms that \tool{} can be used to effectively verify side-channel security guarantees provided by open-source processors in less than 25 hours for our most complex targets. 
Our experiments also show that checking microarchitectural contract satisfaction (as enabled by our decoupling theorem) rather than on top of an architectural reference model significantly speeds up verification (less than 2 hours versus 33 hours for a simple 2-stage processor), allowing us to scale verification to realistic processors. %

\paragraph{Bonus material}
The \tool{} verification tool is available at~\cite{leave-artifact}.
\onlyShortVersion{
An extended version of this paper containing the full formal model and proofs of technical results is available at~\cite{techReport}.
}

\newcommand{\exuarch}{\var{P}}
\newcommand{\exarch}{\var{R}}
\newcommand{\exatk}{\var{A}}
\newcommand{\excontract}{\var{L}}

\renewcommand{\exuarch}{\lstinline{IMPL}}
\renewcommand{\exarch}{\lstinline{ISA}}
\renewcommand{\exatk}{\lstinline{ATK}}
\renewcommand{\excontract}{\lstinline{LM}}

\section{Overview}\label{sec:overview}

\begin{figure}
    \begin{lstlisting}[style={verilog-style}]
    module ISA(input clk, output register);
        
    reg [31:0]  imem [31:0], pc, register;
    wire [31:0] instr = imem[pc]; 
    
    assign op  = instr[7:0];  (* \label{code:spec-dec-start}*)
    assign imm = instr[31:8]; (* \label{code:spec-dec-end}*)
    
    always @(posedge clk) begin 
        pc <= pc + 1;	(* \label{code:spec-pc}*)
    end

    always @ (posedge clk) begin  (* \label{code:spec-ex-start}*)
        case(op)
        `ADD : register <= register + imm;
        `MUL : register <= register * imm;
        `CLR : register <= 0;
    end (* \label{code:monitor-end}*)
    \end{lstlisting}
    \vspace{-10pt}
    \caption{ISA reference model for our running example. \label{fig:reference}}
    \end{figure}

\begin{figure}
    \begin{lstlisting}[style={verilog-style}]
    module IMPL(input clk, output ready, register);
    reg [31:0] imem [31:0], pcF, register;
    
    // Decode
    wire [31:0] instr = imem[pcF];
    
    always @(posedge clk) begin (* \label{code:decode-start}*)
        ex_op  <= inst[7:0]; 
        ex_imm <= inst[31:8];
    end (* \label{code:decode-end}*)

    always @(posedge clk) begin (* \label{code:pc-update-start}*)
      if (ready) pcF <= pcF + 1; (* \label{code:pc-stall}*) 
    end (* \label{code:pc-update-end}*)
    
    assign pc = pcF-2; (* \label{code:spec-arch-pc}*) // Architectural pc

    // Execute 	
    assign ready = (!mult); (* \label{code:ready-flag}*)
    assign rd = we ? wb_res : register;// Forwarding (* \label{code:forwarding}*) 

    log_time_mult(mult, m_imm, m_rd, m_res, done);(* \label{code:mult-by-addition-module}*)
    
    always @ (posedge clk) begin  (* \label{code:execute-start}*) 
      if (ready)  
        case(ex_op)
          `ADD : wb_res <= rd + ex_imm; (* \label{code:addition}*)
                 we <= 1; mult <= 0; 
          `MUL : mult <= 1; we <= 0; (*\label{code:multiply-start}*)
                 m_rd <= rd; m_imm <= ex_imm; (* \label{code:addition-set-mul}*)          
          `CLR : we <= 1; mult <= 0; wb_res <= 0; (*\label{code:clear-start}*)     
      if (done) 
        mult <= 0; wb_res <= m_res; we <= 1;(* \label{code:multiply-set-zero}*) 
    end (* \label{code:execute-end}*) (*\label{code:multiply-end}*)
    
    // Write back
    always @ (posedge clk) begin (*\label{code:begin-writeback}*)
      if (we) // write enabled
        register <= wb_res; retired <= 1; (* \label{code:writeback}*)
      else 
        retired <= 0;
    end(*\label{code:end-writeback}*)
    \end{lstlisting}
    \vspace{-10pt}
    \caption{A simple processor that performs addition and multiplications. 
    The multiplication module \var{log_time_mul} leaks part of both  register value and immediate operand via timing.
    \label{fig:running}}
    \end{figure}

    Here, we illustrate the key points of our approach with an example.
We start by presenting a simple instruction set and the processor implementing it (\Cref{sec:overview:processor}).
Next, we show how microarchitectural leaks can be formalized using leakage contracts (\Cref{sec:overview:contracts}).
Finally, we illustrate how the \tool{} verification tool verifies that the contract is satisfied, thereby ensuring the absence of unwanted leaks (\Cref{sec:overview:verification}).
    
\subsection{A simple processor}\label{sec:overview:processor}
Next, we present the instruction set and processor implementation. %

\begin{figure}
    \begin{subfigure}[t]{0.3\columnwidth}
        \centering
         \begin{equation*}
             \begin{array}[c]{@{}|l@{\;\;}l|@{}}
                \cline{1-2} \rule{0pt}{3ex}
                \var{ADD} \;  2 & \var{MUL} \; 2 \\[1ex]
                \cline{1-2} \rule{0pt}{3ex}
                \var{ADD} \; 2 & \var{ADD} \; 2 \\[1ex]               
                \cline{1-2}
             \end{array}
         \end{equation*}
        \vspace{-5pt}
        \caption{}
        \label{fig:leak:a}
      \end{subfigure}
      \begin{subfigure}[t]{0.3\columnwidth}
        \centering
        \begin{equation*}
        \begin{array}[c]{@{}|l@{\;\;}l|@{}}
            \cline{1-2} \rule{0pt}{3ex}
            \var{ADD} \;  2 & \var{MUL} \; 1 \\[1ex]
            \cline{1-2} \rule{0pt}{3ex}
            \var{ADD} \; 2 & \var{MUL} \; 7  \\[1ex]
            \cline{1-2}              
         \end{array}
        \end{equation*} %
        \vspace{-5pt}
        \caption{\label{fig:leak:b}}
      \end{subfigure}
      \begin{subfigure}[t]{0.3\columnwidth}
        \centering
        \begin{equation*}
        \begin{array}[c]{@{}|l@{\;\;}l|@{}}
            \cline{1-2} \rule{0pt}{3ex}
            \var{ADD} \;  10 & \var{MUL} \; 2 \\[1ex]
            \cline{1-2} \rule{0pt}{3ex}
            \var{ADD} \;  2 & \var{MUL} \; 2 \\[1ex]
            \cline{1-2}               
         \end{array}
        \end{equation*}%
        \vspace{-5pt}
        \caption{\label{fig:leak:c}}
      \end{subfigure} 
      \vspace{-10pt}
\caption{Traces that leak via timing.\label{fig:leak}} 
\end{figure} 

\begin{figure}[ht]
    \begin{tabular}{rl}
    \parbox[c]{1pt}{
            \subcaption{\label{fig:safe:a}}
            }
&
             \parbox[l]{.4\linewidth}{\begin{equation*}
                    \begin{array}[c]{@{}|l@{\quad \;}l@{\quad \;}l@{\quad}|}
                        \cline{1-3} \rule{0pt}{3ex}
                        \var{ADD} \; 5 & \var{CLR} & \var{ADD} \; 4   \\[1ex]
                        \cline{1-3} \rule{0pt}{3ex}
                        \var{ADD} \; 11 & \var{ADD} \; 4 & \var{ADD} \; 9 \\[1ex]
                        \cline{1-3}
                    \end{array}
                \end{equation*}
            }
\\
\parbox[c]{1pt}{
    \subcaption{\label{fig:safe:b}}
    }
&
    \parbox[l]{.4\linewidth}{
        \begin{equation*}
        \begin{array}[c]{|@{\;\;}l@{\quad \;}l@{\quad \;}l@{\quad \;}l@{\quad}|}
            \cline{1-4} \rule{0pt}{3ex}
            \var{ADD} \; 1 & \var{CLR} & \var{ADD} \; 1     &\var{MUL} \; 3 \\[1ex]
            \cline{1-4} \rule{0pt}{3ex}
            \var{ADD} \; 5 & \var{CLR} & \var{ADD} \; 1 &\var{MUL} \; 3 \\[1ex]
            \cline{1-4}                
        \end{array}
        \end{equation*}
    }
    \\
    \parbox[c]{1pt}{
        \subcaption{\label{fig:safe:c}}
        }
&
        \parbox[l]{.4\linewidth}{
            \begin{equation*}
            \begin{array}[c]{|@{\;\;}l@{\quad \;}l@{\quad \;}l@{\quad}|}
                \cline{1-3} \rule{0pt}{3ex}
                \var{ADD} \;  10 & \var{ADD} \;  2  & \var{MUL} \; 2 \\[1ex]
                \cline{1-3} \rule{0pt}{3ex}
                \var{ADD} \;  7  & \var{ADD} \;  5 &  \var{MUL} \; 2 \\[1ex]             
                \cline{1-3}
            \end{array}
            \end{equation*}
        }
    \end{tabular}
\vspace{-10pt}
\caption{Traces that do not leak via timing. \label{fig:safe}}
\end{figure}

\mypara{Instruction set}
We consider an instruction set supporting addition and multiplication of immediates to a single register. 
Instructions consist of the instruction type (\var{ADD}, \var{MUL}, or \var{CLR}) and an
immediate value \var{imm}. 
\var{ADD} adds the immediate to the register value, whereas  \var{MUL} multiplies the register value by the immediate.
Finally, \var{CLR} resets the register to zero.

\Cref{fig:reference} depicts a Verilog reference model \exarch{} for our instruction set that executes one instruction per cycle.
Instructions are stored in the instruction memory \var{imem}.
\Cref{code:spec-dec-start,code:spec-dec-end} decode the instruction into operator (\var{ADD}, \var{MUL}, or \var{CLR}) and operand (immediate value).
Lines \ref{code:spec-ex-start} to \ref{code:monitor-end} case-split on the type of operation and update the register with the new value. 
Finally, \Cref{code:spec-pc} advances the program counter.\looseness=-1

\mypara{Pipelined implementation}
\Cref{fig:running} shows an implementation \exuarch{} of the instruction set that processes instructions in a three-stage pipeline.
If the pipeline is not stalled (flag \var{ready}), the processor starts by fetching a new instruction in line~\ref{code:pc-stall}.
As in \Cref{fig:reference}, the decode stage (lines \ref{code:decode-start} to \ref{code:decode-end}) decodes a new instruction into operator and immediate.
Next, the execute stage executes the decoded instruction (lines \ref{code:execute-start} to \ref{code:execute-end}).
The write-back stage updates the register with the result of the computation (lines \ref{code:begin-writeback} to \ref{code:end-writeback}).
This step is controlled by the write-enabled flag \var{we}. 
Finally, the processor performs forwarding from the execute to the write-back stage  (\cref{code:forwarding}).

Both \var{ADD} (\cref{code:addition}) and \var{CLR} (\cref{code:clear-start}) instructions are executed in a single cycle and their results are passed to the write-back stage. 

In contrast, \var{MUL} instructions (\cref{code:multiply-start} to \cref{code:multiply-end}) may take multiple cycles.
Multiplication starts in \cref{code:multiply-start} by setting register \var{mult} to $1$. 
This indicates that the processor cannot fetch new instructions (\cref{code:ready-flag}) and must stall the pipeline (\cref{code:pc-stall}). 
The processor then multiplies immediate and register value. 
This step is implemented in module \var{log_time_mult} (\cref{code:mult-by-addition-module}),
which we omit. 
The module takes time proportional to the logarithm of \var{m_rd}'s value, \ie{} the register value, to perform the multiplication.\footnote{This timing profile is similar to the Slow Multi-Cycle Multiplier from~\cite{ibex-mult}.}
It also has a fast path that completes the multiplication in a single cycle whenever 
operand or register are \var{0} or \var{1}.
Once multiplication terminates, \var{mul_res} contains the multiplication result and the processor stops stalling the pipeline by setting \var{mult} to $0$ (\cref{code:multiply-set-zero}) and passes the result to the write-back stage.\looseness=-1

\subsection{Specifying side-channel leakage}\label{sec:overview:contracts}
We now illustrate how to use leakage contracts to capture side-channel security guarantees for our example processor.

\mypara{Leakage}
To use the processor from~\Cref{fig:running} securely, we need to know what the
processor may leak to an attacker.
In the following, we consider an attacker that observes the value of the output-ready flag \var{ready} at each cycle, \ie{} it observes the pipeline's timing.

Assume that initially the register has value \var{0}. 
\Cref{fig:leak} shows pairs of instruction sequences that an attacker can distinguish. 
The sequences in \Cref{fig:leak:a} are distinguishable since the upper trace performs a multiplication while the lower trace does not, resulting in a timing difference. 
Similarly, the attacker can distinguish the traces in \Cref{fig:leak:b}, as the upper trace profits from the fast path in the multiplier, while the lower trace does not.
Even though the immediate operands to \var{MUL} are the same in \Cref{fig:leak:c}, the attacker can tell the sequences apart, as the register values are different.

In contrast, \Cref{fig:safe} shows pairs of instruction sequences that are indistinguishable for our 
attacker. 
\Cref{fig:safe:a} does not leak as it does not perform multiplication. 
\Cref{fig:safe:b} initially performs additions with different values, but resets the register state via \var{CLR} before \var{MUL}. 
Finally, \Cref{fig:safe:c} performs additions with different values that result in the same register state before \var{MUL}.

Next, we show how to capture leakage using monitors, which we use to formalize leakage contracts and attackers.

\mypara{Capturing leakage via monitors}
To use the processor securely, we need to distinguish program behaviors that leak from those that do not.
For this, we compose the reference model~\exarch\ (\Cref{fig:reference}), which captures the \emph{functional} behavior of the ISA, with a \emph{leakage monitor} \excontract\ shown below. 
The leakage monitor captures which information may be leaked upon executing instructions. 
The monitor takes as input a module $M$ representing the underlying circuit. 
We denote by \compose{\excontract}{$M$} the composition of \excontract\ and $M$ such that the composition hides $M$'s outputs, and \excontract\ can refer to (but cannot modify) $M$'s internal variables (\Cref{sec:model:leakage-contracts}).

In our example, the monitor leaks whether the operation that is performed is a multiplication or not (\var{ismul}). 
Whenever a \var{MUL} is executed, the monitor additionally leaks the register value (\var{r}) and whether the immediate is 0 or 1 (\var{isFP}), thereby capturing the leaks associated with the multiplier's fast path.
\begin{lstlisting}[style={verilog-style}]
monitor LM(module M, output leak)
  assign inst = M.imem[M.pc];
  assign r = M.register;
  assign op = inst[7:0];
  assign imm = inst[31:8];
  assign isFP = (imm==0 || imm==1);
  assign ismul = (op==`MUL);

  always @( * ) begin
    if (ismul) (* \label{code:monitor-start}*)
        leak = {r, isFP, ismul}; 
    else 
        leak = {0, 0, ismul};
  end
\end{lstlisting}
Note that $\{a, b, c\}$ is Verilog notation for the concatenation of signals $a$, $b$, and $c$.
Consider the leakage observations (\ie{} the values for \var{leak}) produced by $\compose{\excontract}{\exarch}$, \ie{} the leakage monitor applied to the reference model.
All pairs of sequences in \Cref{fig:safe} produce the same observations, whereas all pairs in \Cref{fig:leak} result in different observation traces. 
For example, in \Cref{fig:safe:a}, $\compose{\excontract}{\exarch}$ produces observations consisting of \var{\{0,0,0\}} for both instruction sequences. 
In contrast, for the second instruction of \Cref{fig:leak:a}, the upper sequence produces observation \var{\{2,0,1\}} but the lower one produces~\var{\{0,0,0\}}.

\newcommand*\circled[1]{\tikz[baseline=(char.base)]{
            \node[shape=circle,draw,inner sep=1pt] (char) {#1};}}

\mypara{Attacker observations}
Next, we define the observations an attacker can make about implementation \exuarch. 
Since we consider an attacker that can observe the timing of the computation, we define another monitor \exatk\ that simply exposes the ready bit.
\begin{lstlisting}[style={verilog-style}]
monitor ATK(module M, output leak)
    always @ ( * ) begin 
        leak = M.ready;
    end
\end{lstlisting}
The composition of attacker and implementation \compose{\exatk}{\exuarch} defines the \emph{actual information} an attacker may learn about the implementation.\looseness=-1

\mypara{Leakage contracts}
The composition \compose{\excontract}{\exarch} of leakage monitor and reference model defines a \emph{leakage contract} at the ISA level.
The contract characterizes leaks at the granularity of the execution of instructions from the instruction set, and it expresses which parts of the computation may be leaked by the hardware. 
For programmers, the contract provides a guideline for writing side-channel free code:
secrets should never influence leakage observations.
In our example, any two program executions that differ only in their secrets (\eg{} the initial register value) must produce indistinguishable traces. %

\mypara{Contract satisfaction}
The implementation $\exuarch$ \emph{satisfies} the contract $\compose{\excontract}{\exarch}$ under the attacker $\exatk$ whenever $\exuarch$ leaks no more than specified by the contract under $\exatk$.
That is, circuit $\compose{\exatk}{\exuarch}$ should leak no more than circuit $\compose{\excontract}{\exarch}$, denoted $\ctrsat{\compose{\excontract}{\exarch}}{\compose{\exatk}{\exuarch}}$.
That is, for any pair of initial architectural states for which \compose{\excontract}{\exarch} produces the same leakage observations, \compose{\exatk}{\exuarch} must produce the same attacker observations.
A formal definition of this relation is provided in \Cref{sec:model:leakage-contracts}.  %
For example, for all pairs of instruction sequences shown in \Cref{fig:safe}, \compose{\exatk}{\exuarch} must produce the same sequence of \var{ready} bits. 
In contrast, for the pairs of sequences in \Cref{fig:leak}, the sequence of ready bits may differ, but it does not have to.
Next, we describe our methodology to check that an implementation satisfies a contract.

\begin{figure*}[t]
\centering
\small
\begin{tikzpicture}[scale=1,every text node part/.style={align=center}]
    \node (topleft) at (13,0) {%
    $\compose{\exatk}{\exuarch}$};
    \node (topright) at (6.5,0) {%
    $\compose{\contract}{\exuarch}$};
    \node (bottomright) at (6.5,2) {%
    $\compose{\contract}{\exarch}$};
    \node (topleftleft) at (3,0) {Microarchitecture\\$\exuarch$};
    \node (bottomleftleft) at (3,2) {Instruction Set Architecture\\$\exarch$};

    \node[draw] (decoupling) at (5.7,1) {\large Decoupling\\(\Cref{thm:contract-sat-preconditions})};

    \node (left) at (3,1) {};
    \node (right) at (8,1) {};

    \path[<->] (8.7,0.1) -- node[above,sloped] {\huge $\Longleftrightarrow$} (8.9,1.1);

    \draw[dashed, thick] (left) -- (decoupling);
    \draw[dashed, thick] (decoupling) -- (8.5,0.65);
    
    \draw (topleft) -- (topright);
    \draw (topleft) -- (bottomright);
    
    \draw (topleft) -- node[below, sloped] {Microarch. Contract Satisfaction\\$\leakorder{\excontract}{\exatk}{\exuarch}{\phi}$ (\Cref{def:leakage-ordering})} (topright);
    \draw (topleftleft) -- node[left] {ISA Compliance\\$\isasat{\exuarch}{\exarch}{\phi}$ (\Cref{def:isa-compliance})} (bottomleftleft);
    \draw (topleft) -- node[above, sloped] {Contract Satisfaction\\$\ctrsat{\compose{\excontract}{\exarch}}{\compose{\exatk}{\exuarch}}$ (\Cref{def:contract-satisfaction})} (bottomright);
\end{tikzpicture}
\vspace{-7pt}
\caption{ISA compliance, contract satisfaction, and microarchitectural contract satisfaction.\label{fig:overview}}
\end{figure*}

\subsection{Verifying contract satisfaction}\label{sec:overview:verification}
Formally verifying contract satisfaction amounts to proving that $\ctrsat{\compose{\excontract}{\exarch}}{\compose{\exatk}{\exuarch}}$ holds.
This requires reasoning about pairs of infinite traces from $\compose{\excontract}{\exarch}$ and $\compose{\exatk}{\exuarch}$ for \emph{all} possible initial memories (including both data and instructions) and \emph{all} possible initial microarchitectural states. %
Beyond reasoning about security, this also implicitly requires to show that \exuarch\ correctly implements the ISA. %
In our example, functional correctness bugs in $\exuarch$ would often also result in contract violations as leakage observations are a function of the architectural state.
For instance, assume an incorrectly implemented \var{CLR} instructions that does not reset the register to~$0$.
Then the traces in \Cref{fig:safe:b} would likely be distinguishable via timing.\looseness=-1

While functional correctness is thus crucial for security, it needs to be verified independently of security concerns.
Indeed, there are many existing approaches~\cite{reid2016end,Huang19, Zeng21,burch1994automatic, khune2010automated,patankar1999formal,jhala12001microarchitecture} for checking ISA compliance. %
One of the contributions of this paper is to show how leakage and functional verification can be decoupled from each other, enabling a clean separation of functional and security concerns.

\mypara{ISA compliance} 
So, what does it mean for the implementation to comply with the ISA?
Intuitively, the implementation should go through the same sequence of architectural states as the reference model.
However, the reference model processes one instruction in each cycle, while the implementation overlaps the execution of multiple instructions and may or may not retire an instruction in any given cycle.
To bridge this gap, a \emph{retirement predicate} captures when the processor retires instructions and thus commits changes to the architectural state. 
A retirement predicate $\phi$ over implementation circuit \exuarch\ must satisfy the following constraint: whenever $\phi$ holds, \exuarch's current architectural state corresponds to a valid architectural state of the reference model, and no changes to the architectural state may occur when $\phi$ does not hold. %
For our example, the architectural variables are \var{pc}, \var{imem}, and \var{register} and $\phi \eqdef (retired = 1)$ is a valid retirement predicate. 
In fact, $\phi$ acts as a witness to the fact that \exuarch\ complies with the ISA defined by the reference model \exarch:
For any initial architectural state, \exarch\ transitions through the same sequence of architectural states as \exuarch\ does upon instruction retirement, \ie{} whenever $\phi$ holds. 
We denote this notion of ISA compliance by $\isasat{\exuarch}{\exarch}{\phi}$. %

\mypara{Decoupling leakage and functional correctness} 
Using the retirement predicate, we are able to decouple leakage from functional verification. 
To this end, we first define a \emph{filtered semantics} (in \Cref{sec:model:language}) that only considers states in which $\phi$ holds.
Since $\isasat{\exuarch}{\exarch}{\phi}$ implies that $\exuarch$'s architectural state matches $\exarch$'s whenever $\phi$ holds, the sequence of architectural states produced by the filtered semantics of~$\exuarch$ with respect to $\phi$ is equal to the sequence of states produced by~$\exarch$, assuming the processor is implemented correctly.\looseness=-1

Based on the filtered semantics, we can define the notion of \emph{microarchitectural contract satisfaction}:
To this end, we apply the leakage monitor $\excontract$ directly to $\exuarch$ and then relate $\compose{\excontract}{\exuarch}$ to $\compose{\exatk}{\exuarch}$, bypassing the reference model:
For all pairs of traces of $\compose{\excontract}{\exuarch}$, if contract observations (filtered using $\phi$) are the same, then $\compose{\exatk}{\exuarch}$'s observations must also be the same.
We denote this relation by $\leakorder{\excontract}{\exatk}{\exuarch}{\phi}$.

Our main theorem, \Cref{thm:contract-sat-preconditions} (in \Cref{sec:verification:decoupling}), states that if $\exuarch$ correctly implements $\exarch$ with respect to the retirement predicate $\phi$, then $\leakorder{\excontract}{\exatk}{\exuarch}{\phi}$ \emph{if and only if} $\ctrsat{\compose{\excontract}{\exarch}}{\compose{\exatk}{\exuarch}}$.
This means that we can analyze contract satisfaction purely based on the implementation $\exuarch$.
Figure~\ref{fig:overview} illustrates the main concepts and their relation. %

\mypara{Verification via inductive invariants}
\tool{}, our verification approach, verifies $\leakorder{\excontract}{\exatk}{\exuarch}{\phi}$ by approximating it via the following safety property:
Any two prefixes of traces that agree on their leakage observations also agree on attacker observations \emph{and} determine each instruction's retirement time.
A challenge in this formulation is that differences in attacker observations may surface before the corresponding differences in leakage observations due to pipelined execution and the fact that leakage observations corresponding to an instruction can only be evaluated upon instruction retirement.
We address this challenge by applying a bounded lookahead to the leakage observations.\looseness=-1

Checking this safety property requires appropriate inductive invariants, which would be tedious to come up with manually, in particular for complex designs.
Thus, we synthesize appropriate invariants from a pool of candidate relational invariants following the classic Houdini algorithm~\cite{flanagan2001houdini}. %

\newcommand{\WireVars}{\mathit{Vars}}
\newcommand{\wires}[1]{#1.W}

\section{Formal model}\label{sec:model}

In this section, we present the key components of our formal model.
We start by introducing $\lang{}$, a simple hardware description language (\Cref{sec:model:language}).
Next, we show how to formalize instruction set architectures and microarchitectures in $\lang{}$ (\Cref{sec:model:arch-and-uarch}).
We conclude by formalizing leakage contracts (\Cref{sec:model:leakage-contracts}).

\subsection{$\lang{}$: A Hardware Description Language}\label{sec:model:language}

$\lang$ is a language for specifying synchronous sequential circuits.
It captures the key features of hardware description languages like Verilog and VHDL, and we use it as the core language for \tool{}.

\mypara{Syntax}
The syntax of $\lang$ is given in \Cref{figure:language:syntax}.
\emph{Expressions} $e$ are built from values $\Val = \Nat \cup \{ \bot\}$, which are natural numbers or the designated value $\bot$, registers $\Var$, which store values, and variables $\WireVars$, which are shorthands for more complex expressions.
Expressions can be combined using unary operators $\unaryOp{e}$, binary operators $\binaryOp{e_1}{e_2}$, if-then-else operators $ \ite{e_1}{e_2}{e_3}$, and bit-selection operators $e_1[e_2:e_3]$.
An \emph{assignment} $\passign{r}{e}$ sets the next value of register $x$ to the
value of expression $e$ in the current cycle.
A \emph{wire} $v = e$ always has the value of expression $e$. %
Finally, a \emph{circuit}~$C$ consists of a set of assignments $A$, a set of wires $W$, and a set of
outputs $O \subseteq \Var \cup \WireVars$.

Given a circuit $C$, we refer to its assignments as $\assignments{C}$, to its wires as $\wires{C}$, and to its outputs as $\outputs{C}$.
The set $\readVars{C}$ of read registers consists of all registers $x$ that occur in at least one right-hand side of an assignment in $C.A$ or a wire in $\wires{C}$.
Similarly, the set $\writeVars{C}$ of write registers consists of all registers $x$ occurring in left-hand sides of assignments in $\assignments{C}$.
Finally, the set $\wireVars{C}$ of wire variables consists of all variables $v$  occurring in left-hand sides of wires in $\wires{C}$.
We assume that (1) $C.O \subseteq \vars{C} \cup \wireVars{C}$, where $\vars{C} = \readVars{C} \cup
\writeVars{C}$, (2) each register and variable is on the left-hand side of at most one
assignment or wire, and (3) wires in $\wires{C}$ do not introduce cyclic dependencies.

\begin{example}
    \label{ex:syntax}
    Consider the circuit~$\SimpISA$ given below.
    The circuit implements a simple ISA, in which instructions consist solely of immediate values $m[pc]$ that are
    retrieved from memory $m$ and added to the single internal register $reg$.\footnote{For simplicity, in the examples we treat memories as addressable arrays.
    For instance, $m[pc]$ denotes the value in $m$ at position $pc$. While this can be desugared in the syntax from \Cref{figure:language:syntax}, we decided against this to simplify our encodings.}
    \begin{equation*}
        \SimpISA = \{\passign{pc}{pc+1}, \passign{reg}{reg+m[pc]}\} : \{\} : \{ reg \}
    \end{equation*}
    We have $\vars{\SimpISA} = \readVars{\SimpISA} = \{pc, reg, m\}$ and $\writeVars{\SimpISA} = \{pc,
    reg\}$, and the single output $reg$; the circuit satisfies our assumptions.
\end{example}

\begin{figure}
	{\centering
	\begin{tabular}{llcl}
	\multicolumn{4}{l}{\bf Basic Types} \\
	\textit{(Registers)} 	&  $r$		& $\in$ & $\Var $ \\
    \textit{(Variables)} 	    &  $v$		& $\in$ & $\WireVars $ \\
    \textit{(Identifiers)}  &  $i$      & $\in$ & $\Var \cup \WireVars$ \\
    \textit{(Values)} 		&  $n$ 		& $\in$ & $\Val = \Nat \cup \{ \bot\}$  \\
	\multicolumn{4}{l}{\bf Syntax} \\
	\textit{(Expressions)} 	&  $e$		& $:=$ & $n \mid i \mid \unaryOp{e} \mid \binaryOp{e_1}{e_2}$ \\
    & & & $\mid \ite{e_1}   {e_2}{e_3} \mid e_1[e_2:e_3]$ \\
    \textit{(Wires)}        &  $w$ 		& $:=$ & $v = e$ \\
                            &  $W$      & $:=$ & $\{ w_1, \ldots, w_k\}$ \\
	\textit{(Assignments)} 	&  $a$ 		& $:=$ & $\passign{r}{e}$ \\
                            &  $A$      & $:=$ & $\{ a_1, \ldots, a_n\}$ \\
    \textit{(Outputs)}      &  $O$      & $:=$ & $\{ i_1, \ldots, i_m\}$ \\
    \textit{(Circuits)}     &  $C$ 		& $:=$ & $A : W : O$ 
	\end{tabular}
	}
    \vspace{-10pt}
\caption{$\lang$ syntax}\label{figure:language:syntax}
\end{figure}

\mypara{Semantics}
We formalize the semantics of $\lang$ circuits by specifying how their state is updated at each cycle.
We model the state of a circuit as a \emph{valuation} $\val$ that maps registers in $\Var$ to values in $\Val$, \ie{}  $\val: \Var \to \Val$.
Given a circuit $C$, $\states{C}$ denotes the set of all possible valuations over $\vars{C}$.
Given a valuation $\val$ and a set of registers $V$, the projection $\project{\val}{V}$ restricts the scope of $\val$ to the registers in $V$, \ie $\project{\val}{V}(x) = \val(x)$ for all $x \in V$ and  $\project{\val}{V}(x) =  \bot$ otherwise.
Finally, given two valuations $\val,\val'$ and a set of registers  $V$, $\val \equivVars{V} \val'$ denotes that $\val$ and $\val'$ agree on the values of all registers in $V$, \ie $\val \equivVars{V} \val'$ iff $\project{\val}{V} = \project{\val'}{V}$.

The \emph{semantics} $\llbracket C\rrbracket$ of a circuit $C$ takes as input a valuation $\val$ and outputs the valuation $\val'$ at the next cycle.
An \emph{execution} for $C$ starting from valuation $\val$ is the infinite sequence of valuations obtained by repeatedly applying $\llbracket C\rrbracket$.
The \emph{infinite trace semantics} $\llbracket C \rrbracket^\infty$ of a circuit $C$ maps each valuation $\val$ to the infinite sequence of valuations for $C$'s outputs, where the $i$-th valuation corresponds to the circuit's output after $i$ cycles.\footnote{With a slight abuse of notation, the trace semantics extend valuations to also record values of wires that are part of a circuit's outputs.}
Additionally, the \emph{filtered infinite trace semantics} $\llbracket C \rrbracket^\infty|\phi$ outputs only the valuations in $\llbracket C \rrbracket|\phi$ that satisfy a given predicate $\phi$ (other valuations are dropped).
Finally, $\hwEval{C}{\mu, i}$ denotes the valuation obtained by executing $C$ for $i$ cycles starting from valuation $\mu$, whereas $C,\mu \models \phi$ denotes that $\phi$ is satisfied for circuit $C$ and valuation $\mu$. 
The full formalization of $\lang$ is given in~\techReportAppendix{app:semantics}.

\begin{example}
    \label{ex:semantics}
    Consider again circuit $\SimpISA$ from \Cref{ex:syntax}.        
    Let us pick an initial valuation $\mu$, such that $\mu(pc)=0$, $\mu(reg)=0$, and
    \begin{equation*}
    \begin{array}[c]{@{}r@{\; \mathop{=} \;}l@{\qquad}l@{}}
        \mu(m)(i) & i & \text{for} \;\; 0 \leq i \leq 10\\
        \mu(m)(i) & 0 & \text{otherwise} \ .
    \end{array}
    \end{equation*}
    Executing a single step gives us $\mu'=\hwEval{\SimpISA}{\mu}$, with $\mu'(pc)=1$, and $\mu'(reg)=0$. 
    Since only the program counter changed, we get $\mu \equivVars{\theSet{reg,mem}} \mu'$, but not $\mu \equivVars{\theSet{pc}} \mu'$. 
    The trace $\hwEvalInfty{\SimpISA}{\mu}$ consists of the following sequence of register values (since the register value does not change after step 11), where $\cdot$ denotes concatenation:
    \begin{align*}
    \hwEvalInfty{\SimpISA}{\mu} = 0 \cdot 0 \cdot 1 \cdot 3 \cdot 6 \cdot 10 \cdot 15 \cdot 21 \cdot 28 \cdot 36 \cdot 45 \cdot 55 \cdot 55 \cdot 55 \ldots
    \end{align*}
    As an example of filtering, consider the predicate $\phi:=pc \; \mathbf{mod} \; 2 = 0$ indicating whether the program counter is even.
    The filtered semantics associated with $\phi$ yields the following sequence:
    \begin{align*}
        \hwEvalInftyFilter{\SimpISA}{\mu}{\phi} = 0 \cdot 1 \cdot 6 \cdot 15 \cdot 28 \cdot 45 \cdot 55 \cdot \ldots 
    \end{align*}
    \end{example}

\subsection{Modeling architectures and microarchitectures}\label{sec:model:arch-and-uarch}

We now show how instruction set architectures (short: architectures) and microarchitectures can be modeled in $\lang{}$.
Then, we formalize what it means for a microarchitecture $\uarch$ to correctly implement an architecture $\arch$.

\mypara{Architectures}
We view architectures as state machines that define how the execution progresses through a sequence of architectural states, where each transition corresponds to the execution of a single instruction.
Given a set of \emph{architectural registers} $\archVars$, we model an \emph{architecture} as a circuit $\arch$ over $\archVars$, \ie{} $\vars{\arch} = \outputs{\arch} = \archVars$.
We assume that a subset $\initStates{\arch}$ of $\arch$'s states are identified as initial states.

\begin{example}
\label{ex:arch}    
Consider again circuit $\SimpISA$ from~\Cref{ex:syntax}. 
Its variables $\vars{\SimpISA} = \{pc, reg, m\}$ form the architectural state of the ISA.
We identify as initial states all valuations $\mu$ such that $\mu(pc)=0$ and $\mu(reg)=0$. 
In the circuit from \Cref{fig:reference} the architectural state is given by $\vars{R} = \theSet{imem, pc, register}$ whereas $instr, op$, and $imm$ are not listed as they are wires.
\end{example}

\mypara{Microarchitectures}
We model microarchitectures as circuits that capture the execution at the granularity of clock cycles.
Thus, a \emph{microarchitecture} is a circuit $\uarch$ that refers to both architectural registers in $\archVars$ and to additional microarchitectural registers $\uarchVars$ such that $\vars{\uarch} = \outputs{\uarch}=\archVars \cup \uarchVars$ and $\archVars \cap \uarchVars = \emptyset$.
We assume that a subset $\initStates{\uarch}$ of $\uarch$'s states is identified as the initial states and require that $\project{\val}{\archVars} \in \initStates{\arch}$ for any state $\mu \in \initStates{\uarch}$, \ie{} the architectural part of an initial microarchitectural state should be an initial architectural state.\looseness=-1 %

\begin{example} \label{ex:muarch}
    Let us look at a microarchitectural implementation $\SimpIMP$ of the ISA in
    example~\ref{ex:syntax}. The implementation, shown below, can be in one of two
    states (indicated by the register $\mathit{st}$): execute state
    ($\mathit{st} = 0$) or write-back state ($\mathit{st} = 1$). In the execute
    state, $\SimpIMP$ computes the result of adding the immediate to the current
    register value and assigns it to the variable $res$; it then moves to the
    write-back state (line 3). In the write-back state, $\SimpIMP$ writes the result to the
    register $reg$, moves the state to the execute stage, and increments the program
    counter (line 4). If the immediate value is zero, the implementation
    triggers a fast path which keeps the circuit in the execute state, increments
    the program counter, and leaves the register unchanged (line 2). Finally,
    the circuit updates the variable $ret$ which indicates whether the circuit
    retired in the current step. For readability, we write the example in an
    extended syntax that allows branches at the assignment level.\footnote{This
    syntax can be easily expanded into the one in \Cref{figure:language:syntax}
    by pushing branches into expressions. For example, we can rewrite
    $\mathbf{if} \; e \; \mathbf{th} \; \{ \; \passign{x}{a} \; \} \;
    \mathbf{el} \; \{ \; \passign{x}{b} \; \}$ as $\passign{x}{\mathbf{if} \; e
    \; \mathbf{th} \; a \; \mathbf{el} \; b \; }$.}
    \begin{equation*}
        \begin{array}[c]{@{}l@{\qquad}l@{}}
       1 & \mathbf{if} \; st=0 \; \mathbf{then} \\
       2 & \quad \mathbf{if} \; m[pc]=0 \; \mathbf{th} \; \{ \passign{st}{0}, \passign{pc}{pc+1}, \passign{ret}{1} \} \\
       3 & \quad \mathbf{el} \; \{ \passign{st}{1}, \passign{res}{m[pc] + reg}, \passign{ret}{0} \} \\
       4 & \mathbf{el} \; \{ \passign{st}{0}, \passign{reg}{res}, \passign{pc}{pc+1}, \passign{ret}{1} \} : \{reg\}  \\
        \end{array}
        \end{equation*}
    In addition to architectural variables $\{pc, reg, m\}$, the implementation
    contains microarchitectural variables $\{ st, res, ret\}$. We pick as our
    initial valuations all $\mu$ such that $\mu(pc)=0$, $\mu(reg)=0$,
    $\mu(st)=0$, and $\mu(ret)=1$. As required, the initial state for
    architectural variables $\{pc, reg, m\}$ agrees with the state from
    \Cref{ex:arch}.
\end{example}

\mypara{ISA compliance}
To correctly implement an architecture $\arch$, an implementation $\uarch$ needs to change the architectural state in a manner consistent with $\arch$.
We capture this with the help of a \emph{retirement predicate} $\phi$, a predicate indicating when $\uarch$ retires instructions.
Then, we say that a microarchitecture $\uarch$ implements an architecture $\arch$ (\Cref{def:isa-compliance}) if one can map changes of the architectural state in $\uarch$ to $\arch$'s executions using $\phi$. %

\begin{definition}\label{def:isa-compliance}
A microarchitecture $\uarch$ \emph{correctly implements} an architecture $\arch$ given a retirement predicate $\phi$  over $\vars{\uarch}$, written $\isasat{\uarch}{\arch}{\phi}$, if for all valuations $\mu \in \initStates{\uarch}$:
\begin{compactenum}
\item \emph{(Witnessed architectural changes agree with $\arch$)} $\hwEvalInftyFilter{\uarch}{\mu}{ \phi } \equivVars{\archVars} \hwEvalInfty{\arch}{ \mu }$, and \label{def:isa-cond1}
\item \emph{(No architectural changes beyond those witnessed)} $\hwEval{\uarch}{\mu, i}\hspace{-0.5mm} \equivVars{\archVars}\hspace{-0.5mm} \hwEval{\uarch}{\mu, i-1}$ whenever $\hwEval{\uarch}{\mu, i} \not\models \phi$.\looseness=-1 \label{def:isa-cond2}
\end{compactenum}
\end{definition}

\noindent
The predicate $\phi$ characterizes when instructions are retired, \ie{} when instructions modify the architectural state.
\Cref{def:isa-compliance} uses~$\phi$ to map architectural changes made by $\uarch$ to single steps in $\arch$'s executions.
This is sufficient for  single-issue processors, which retire at most one instruction per cycle.
Multiple-issue processors, which may retire multiple instructions in a single cycle, require more complex ways of mapping architectural changes made by $\uarch$ to $\arch$'s steps.
To simplify our model, we decided against more complex ISA compliance criteria since \tool{}'s verification approach (\Cref{sec:verification}) is decoupled from ISA compliance.
\begin{example}
\label{ex:mu-arch-comp}
Let's again consider implementation circuit $\SimpIMP$ from \Cref{ex:muarch}.
We choose as retirement predicate $\phi \eqdef ret=1$. 
Let's consider again valuation $\mu$ from  \Cref{ex:semantics}, which
maps $pc=0$, and $\mu(m)(i)=i$, for $0 \leq i \leq 10$.
Running $\SimpIMP$ on $\mu$ from produces the following sequence of register values,
where we underline a register value whenever $\phi$ holds on the corresponding
state. 
\begin{align*}
    \hwEvalInfty{\SimpIMP}{\mu} = \underline{0} \cdot \underline{0} \cdot 0 \cdot \underline{1} \cdot 1 \cdot \underline{3} \cdot 3 \cdot \underline{6} \cdot 6 \cdot \underline{10} \cdot 10 \cdot \underline{15} \cdot 15 \cdot \ldots 
\end{align*}
It's easy to check that $\hwEvalInftyFilter{\SimpIMP}{\mu}{\phi}$, \ie the sequence of
underlined values, matches $\hwEvalInfty{\SimpISA}{\mu}$, and that the register value
remains unchanged whenever $\phi$ doesn't hold. Since this is true, not only for
$\mu$ but for all valid initial states, we can conclude that $\SimpIMP$ correctly
implements $\SimpISA$, \ie  $\isasat{\SimpIMP}{\SimpISA}{\phi}$.

\end{example}

\subsection{Leakage contracts}\label{sec:model:leakage-contracts}

In this section, we first introduce monitoring circuits, which we use to specify leakage contracts and attackers.
Then, we formalize contract satisfaction~\cite{contracts2021} within our modeling framework.

\mypara{Monitoring circuits}
Monitoring circuits \emph{monitor} the behavior of another circuit, and we will use them to formalize leakage contracts and attackers.
We say that circuit $M$ is a \emph{monitoring circuit} for circuit~$C$ if 
\begin{inparaenum}
\item $\writeVars{C} \cap \writeVars{M} = \emptyset$, \ie{} the two circuits write to separate sets of registers, 
\item $\wireVars{C} \cap \wireVars{M} = \emptyset$, \ie the two circuits write to separate wire variables, and
\item $\vars{C} \cap \writeVars{M} = \emptyset$, \ie{} $M$ does not influence $C$'s behavior. %
\end{inparaenum}
Additionally, $M$ is \emph{combinatorial} whenever $\readVars{M} \subseteq \vars{C}$, \ie{} $M$ only reads from $C$ variables and thus does not have state of its own.
Finally, the \emph{composition} of the monitoring circuit $M$ and the monitored circuit $C$, written $\compose{M}{C}$, is  the circuit defined as $C.A \cup M.A : \wires{C} \cup \wires{M} : M.O$, which computes over $C$'s state without changing its behavior.

\mypara{Leakage contracts}
A \emph{leakage contract} is the composition of a leakage monitor $\contract$, \ie{} a combinatorial monitoring circuit $\contract$
for the architecture $\arch$, with the architecture $\arch$ itself.
That is, a leakage contract $\compose{\contract}{\arch}$ discloses parts of the architectural state during $\arch$'s execution at the granularity of instruction execution.

\mypara{Hardware attackers}
We formalize an \emph{attacker} as a combinatorial monitoring circuit $\atk$ for the microarchitecture $\uarch$.
That is, an attacker observes parts of the microarchitecture's state during the execution at the granularity of clock cycles.
\begin{example}
    Consider again circuit~$\SimpISA$ from~\Cref{ex:syntax}, the ISA specification of our running example. 
    We define the leakage monitor~$\SimpLM$, which leaks whether the current instruction is zero. 
    As $\SimpLM$ only reads $\SimpISA$'s variables, it is combinatorial.
    \begin{equation*}
        \SimpLM = \{\} : \{ v = (m[pc]=0) \}: \{ v \}
    \end{equation*} 
    Consider again the valuation $\mu$ from \Cref{ex:semantics}, which maps
    $\mu(m)(i)=i$, for $0 \leq i \leq 10$. Since for $i \leq 10$, only the first
    instruction is zero, executing $\compose{\SimpLM}{I}$ yields the following
    sequence.
    \begin{equation*}
    \hwEvalInfty{\compose{\SimpLM}{\SimpISA}}{\mu} = 1 \cdot 0 \cdot 0 \cdot 0 \cdot \dots \ .
    \end{equation*}
\end{example}
\begin{example}
    Next, consider the implementation circuit $\SimpIMP$ from \Cref{ex:muarch}. We
    define the following attacker monitor, which leaks the  program
    counter and thus the timing of the computation.
        \begin{equation*}
            \SimpATK = \{\} : \{\}: \{ pc \}
    \end{equation*} 
    Running $\compose{\SimpATK}{\SimpIMP}$ on $\mu$ yields the following sequence.
    \begin{equation*}
        \hwEvalInfty{\compose{\SimpATK}{\SimpIMP}}{\mu} = 0 \cdot 1 \cdot 1 \cdot 2 \cdot 2 \cdot 3 \cdot 3 \cdot \dots \ .
    \end{equation*}
\end{example}
\mypara{Contract satisfaction}
\Cref{def:contract-satisfaction} formalizes the notion of contract satisfaction~\cite{contracts2021}.
Intuitively, a microarchitecture $\uarch$ {satisfies} the contract $\compose{\contract}{\arch}$ for an attacker $\atk$ if $\atk$ cannot learn more information about the initial architectural state by monitoring $\uarch$'s executions than what is exposed by $\compose{\contract}{\arch}$.
That is, for any two initial states that agree on their microarchitectural part\footnote{Following~\cite{contracts2021}, we assume that secrets initially reside only in the architectural state and that attackers can observe the initial values of registers in $\uarchVars$, \ie{}  $\val \equivVars{\uarchVars} \val'$.}, whenever $\compose{\contract}{\arch}$ results in identical traces, then $\compose{\atk}{\uarch}$ also results in identical traces (\ie{} $\atk$ cannot distinguish the two initial architectural states).

\begin{definition}\label{def:contract-satisfaction}
Microarchitecture $\uarch$ \emph{satisfies} contract $\compose{\contract}{\arch}$ for attacker $\atk$, written $\ctrsat{ \compose{\contract}{\arch} }{ \compose{\atk}{\uarch} }$, if for all valuations $\val,\val' \in \initStates{\uarch}$ such that $\val \equivVars{\uarchVars} \val'$,
    if  $\hwEvalInfty{ \compose{\contract}{\arch} }{ \val  }  = \hwEvalInfty{ \compose{\contract}{\arch} }{  \val' }$, 
    then $\hwEvalInfty{ \compose{\atk}{\uarch} }{ \val } = \hwEvalInfty{ \compose{\atk}{\uarch} }{ \val' }$.
\end{definition}

We remark that \Cref{def:contract-satisfaction} 
 refers to 4 different traces: two contract traces from $\compose{\contract}{\arch}$ and two attacker traces from $\compose{\atk}{\uarch}$.

\begin{figure}
    \hspace{5mm}
    \begin{subfigure}[t]{0.4\linewidth}  
        \parbox[c]{.001\linewidth}{
        \subcaption{\label{fig:running:a}}
        }
        \hspace{-.1\linewidth}      
        \parbox[c]{.22\linewidth}{
        \begin{equation*}
            \begin{array}[c]{@{}|l@{\quad \;}l@{\quad \;}l@{\quad}|}
                 \cline{1-3} \rule{0pt}{3ex}
                1 & 0 & 2   \\[1ex]
                \cline{1-3} \rule{0pt}{3ex}
                5 & 1 & 3 \\[1ex]
           \cline{1-3}
            \end{array}
        \end{equation*}
        \hfill~
        }
    \end{subfigure}
    \begin{subfigure}[t]{0.4\linewidth}  
        \parbox[c]{.001\linewidth}{
        \subcaption{\label{fig:running:b}}
        }
        \hspace{-.1\linewidth}      
        \parbox[c]{.32\linewidth}{    
        \begin{equation*}
            \begin{array}[c]{@{}|l@{\quad \;}l@{\quad \;}l@{\quad}|}
                \cline{1-3} \rule{0pt}{3ex}
                1 & 0 & 2   \\[1ex]
                \cline{1-3} \rule{0pt}{3ex}
                5 & 0 & 3 \\[1ex]
                \cline{1-3}
            \end{array}
        \end{equation*}
        }
        \hfill~
    \end{subfigure}
    \vspace{-10pt}
\caption{Two pairs of instruction traces. \label{fig:run:traces}}
\end{figure}

\begin{example}
    \label{ex:run-contract-sat}
    Let's consider the two pairs of memories (a) and (b) shown in
    \Cref{fig:run:traces}. We will check contract satisfaction, \ie that
    $\ctrsat{ \compose{\SimpLM}{\SimpISA}}{ \compose{\SimpATK}{\SimpIMP}}$ on these particular traces.
    
    Let us start with the instructions from Figure \ref{fig:running:a}. Consider two states
    $\mu_a$ and $\mu_a'$, such that $\mu_a(m)$ contains the upper 
    instructions in \Cref{fig:run:traces}, and $\mu_a'(m)$ contains the lower
    ones. For $i \geq 3$, we let $\mu_a(m)=\mu_a'(m)=0$. 
    Running $\mu_a$ and $\mu_a'$ on the contract, we get:
    \begin{align*}
        \hwEvalInfty{\compose{\SimpLM}{\SimpISA}}{\mu_a} &= 0 \cdot \underline{1} \cdot 0 \cdot 1 \cdot 1 \cdot 1 \cdot \dots \\
        \hwEvalInfty{\compose{\SimpLM}{\SimpISA}}{\mu_a'} &= 0 \cdot \underline{0} \cdot 0 \cdot 1 \cdot 1 \cdot 1 \cdot \dots      
    \end{align*}
    As the contract traces differ in the second position, contract satisfaction
    holds trivially. Next, consider the traces in Figure \ref{fig:running:b}. As
    before, we construct valuations $\mu_b$ for the upper trace, and $\mu_b'$
    for the lower trace. We get the traces below.
    \begin{equation*}
        \hwEvalInfty{\compose{\SimpLM}{\SimpISA}}{\mu_b} = \hwEvalInfty{\compose{\SimpLM}{\SimpISA}}{\mu_b'} = 0 \cdot 1 \cdot 0 \cdot 1 \cdot 1 \cdot 1 \cdot \dots \ .    
    \end{equation*}
    As both valuations produce the same trace, we need to check the attacker
    observations on the implementation. We get 
    \begin{equation*}
        \hwEvalInfty{\compose{\SimpATK}{\SimpIMP}}{\mu_b} = \hwEvalInfty{\compose{\SimpATK}{\SimpIMP}}{\mu_b'} = 0 \cdot 0 \cdot 1 \cdot 2 \cdot 2 \cdot 3 \cdot 4 \cdot 5 \cdot \dots \ .    
    \end{equation*}
    We can therefore conclude that contract satisfaction holds for these traces. To verify contract satisfaction, we need to not only check this property for $\mu_b$ and $\mu_b'$, but for any pair of traces.
    We will discuss our approach for this in the next section.
\end{example}

\section{Verifying contract satisfaction}\label{sec:verification}

Here, we present our verification approach for checking contract satisfaction.
First, we introduce a decoupling theorem that allows us to separate security and functional correctness proofs (\Cref{sec:verification:decoupling}).
Next, we present (and prove sound) an algorithm for verifying microarchitectural contract satisfaction (\Cref{sec:verification:algorithm}).
All proofs are in \techReportAppendix{app:proofs}.\looseness=-1

\subsection{Decoupling contract satisfaction from ISA}\label{sec:verification:decoupling}

Since a leakage contract $ \compose{\contract}{\arch}$ is defined on top of $\arch$, proving contract satisfaction according to \Cref{def:contract-satisfaction} requires 
reasoning about security \emph{and} functional compliance with respect to $\arch$ (since one needs to map contract traces from  $ \compose{\contract}{\arch}$ to implementation traces). %
We address this challenge by decoupling reasoning about security and about $\arch$ compliance.

\mypara{Leakage ordering}
For this, we start by introducing a leakage ordering between combinatorial monitoring circuits for an underlying circuit $C$.
Intuitively, a monitor $M$ for $C$ ``leaks less'' (\ie{} exposes less information) than another monitor $M'$ for $C$ if whenever $\compose{M'}{C}$ produces equivalent traces on two initial states, then $\compose{M}{C}$ also produces equivalent traces.
\Cref{def:leakage-ordering} formalizes this concept and extends it to support the filtered semantics. 

\begin{definition}\label{def:leakage-ordering}
    Monitor \emph{$M'$ leaks at most as much information as} monitor $M$ about circuit $C$, given registers $V \subseteq \vars{C}$, and predicate $\phi$ (over $C$), written $\leakorder{M}{M'}{C}{V,\phi}$, if 
    for all valuations $\val,\val' \in \initStates{C}$ such that $\val \equivVars{V} \val'$, if $\hwEvalInftyFilter{\compose{M}{C}}{\mu}{\phi} = \hwEvalInftyFilter{\compose{M}{C}}{\mu'}{\phi}$, then $\hwEvalInfty{\compose{M'}{C}}{ \mu } = \hwEvalInfty{\compose{M'}{C}}{ \mu'}$.
\end{definition}
Differently from \Cref{def:contract-satisfaction} (which is defined in terms of four traces), \Cref{def:leakage-ordering} 
 is defined in terms of only two traces of $C$.
\begin{example}
    We can use our new definition to express contract satisfaction over the
    implementation only, using predicate~$\phi$. Consider again the two pairs of
    traces in \Cref{fig:run:traces} from \Cref{ex:run-contract-sat}. If we
    assume that the implementation is functionally correct, that is, it
    satisfies \Cref{def:isa-compliance}, we can replace the specification $\SimpISA$
    by its implementation $\SimpIMP$. In particular, since \Cref{def:isa-compliance}
    ensures that $\SimpISA$'s architectural values match $\SimpIMP$'s whenever
    retirement predicate $\phi = (ret=1)$ holds, we can check contract satisfaction by
    checking $\leakorder{ \SimpLM }{\SimpATK}{\SimpIMP}{\{ \mathit{st},
    \mathit{res}, \mathit{ret} \}, \phi}$. We call this condition
    \emph{microarchitectural contract satisfaction}.
    Let us now check this property for the traces in \Cref{fig:running:b}.
    Running $\compose{\SimpLM}{\SimpIMP}$, we get the following, where we
    underline outputs whenever $\phi$ holds.
    \begin{equation*}
        \hwEvalInfty{\compose{\SimpLM}{\SimpIMP}}{\mu_b} = \hwEvalInfty{\compose{\SimpLM}{\SimpIMP}}{\mu_b'} = \underline{0} \cdot 0 \cdot \underline{1} \cdot \underline{0} \cdot 0 \cdot \underline{1} \cdot \underline{1} \cdot \dots \ .    
    \end{equation*}
    This means the premise of the implication is satisfied, and we need to check the conclusion. As before, we get
    \begin{equation*}
        \hwEvalInfty{\compose{\SimpATK}{\SimpIMP}}{\mu_b} = \hwEvalInfty{\compose{\SimpATK}{\SimpIMP}}{\mu_b'} = 0 \cdot 0 \cdot 1 \cdot 2 \cdot 2 \cdot 3 \cdot 4 \cdot 5 \cdot \dots \ .    
    \end{equation*}
    which establishes $\leakorder{ \SimpLM }{ \SimpATK }{\SimpIMP}{\{ \mathit{st}, \mathit{res}, \mathit{ret}\},
    \phi}$ for $\mu_b$ and $\mu_b'$. We formalize this idea in \Cref{thm:contract-sat-preconditions}.
\end{example}

\mypara{Decoupling theorem}
\Cref{thm:contract-sat-preconditions} states that, for functionally correct processors, \emph{microarchitectural contract satisfaction} (\ie{} $\leakorder{ \contract }{ \atk }{\uarch}{\uarchVars, \phi}$, which only refers to the microarchitecture $\uarch$), is equivalent to contract satisfaction (\Cref{def:contract-satisfaction} which refers to  architecture $\arch$ and microarchitecture $\uarch$).
This allows us to cleanly separate reasoning about security and about functional correctness (without losing precision).
In particular, we can split proving contract satisfaction into proving microarchitectural contract satisfaction (which ensures the absence of leaks with respect to $\uarch$) and ISA compliance.
\tool{} leverages \Cref{thm:contract-sat-preconditions} to only reason about security, whereas ISA compliance can be verified separately using techniques focusing on functional correctness~\cite{reid2016end}.

\begin{theorem}[Decoupling Theorem]\label{thm:contract-sat-preconditions}
If $\isasat{\uarch}{\arch}{\phi}$ holds for retirement predicate $\phi$, then
\[ \leakorder{ \contract }{ \atk }{\uarch}{\uarchVars, \phi} \Leftrightarrow \ctrsat{ \compose{\contract}{\arch}}{ \compose{\atk}{\uarch}}. \]
\end{theorem}

\subsection{Verifying microarchitectural contract satisfaction}\label{sec:verification:algorithm}

In this section, we present an algorithm for checking microarchitectural contract satisfaction, \ie{} $\leakorder{ \contract }{ \atk }{\uarch}{\uarchVars, \phi}$.
We first introduce notation for formalizing our verification queries in terms of temporal logic formulas.
Next, we present the verification algorithm and conclude by proving its soundness.

\mypara{Notation}
To formalize our verification queries, we use a linear temporal logic over $\lang{}$ circuits.
Formulas $\Phi$ in this logic are constructed by combining $\lang$ predicates $\phi$ with temporal operators $\future$ (denoting ``in the next cycle''), $\boundedFuture{B}$ (denoting ``for the next $B$ cycles''), and $\alwaysFuture$ (denoting ``always in the future''), and the usual boolean operators.
Given a temporal formula $\Phi$ over a circuit $C$, we write $C, \mu, i \models \Phi$ to denote that the formula is satisfied for initial state $\mu$ at cycle $i$.
We write $C, \mu \models \Phi$ to mean $C, \mu, 0 \models \Phi$, and $C \models \Phi$ to mean that $C, \mu \models \Phi$ holds for all $\mu$.
Our temporal logic is standard; we provide its formalization  in 
\techReportAppendix{app:temporal-logic}.
\begin{example}
    Consider again circuit $\SimpISA$ from \Cref{ex:syntax}. Using initial valuation $\mu$, where $\mu(pc) = 0$, the following holds.
    \begin{equation*}
        \begin{array}[t]{@{}l@{\qquad}l@{}}
            \SimpISA, \mu \models pc = 0 & \SimpISA, \mu \models \future(pc = 1) \\    
            \SimpISA, \mu \models \boundedFuture{3} (pc \leq 3) & \SimpISA \models pc \geq 0 \to \alwaysFuture (pc \geq 0) 
        \end{array}
    \end{equation*}         
\end{example}

\mypara{Product circuit}
Verifying microarchitectural contract satisfaction requires us to reason about \emph{pairs of executions} of $\uarch$, \ie it is a 2-hyperproperty~\cite{clarkson2010hyperproperties}.
We transform hyperproperties into properties over a single execution using a construction called \emph{self-composition}~\cite{barthe2011secure}.
For this, we construct a \emph{product circuit} that executes two copies of a circuit $C$ ($\runA{C}$ and $\runB{C}$) in parallel.
Given circuit $C = \{ \passign{x_1}{e_1}, \ldots, \passign{x_n}{e_n} \} : \{ v_1 = e'_{1}, \ldots, v_k = e'_{k} \} : {o_1, \ldots, o_m}$, we define its \emph{product circuit} $\product{C}{C}$ as $\{ \passign{\runA{x_1}}{\runA{e_1}}, \ldots, \passign{\runA{x_n}}{\runA{e_n}}, \passign{\runB{x_1}}{\runB{e_1}}, \ldots, \passign{\runB{x_n}}{\runB{e_n}} \} : \{ \runA{v_1} = \runA{e'_1}, \runB{v_1} = \runB{e'_1}, \ldots  \runA{v_k} = \runA{e'_k}, \runB{v_k} = \runB{e'_k},  \} : \{\runA{o_1}, \ldots, \runA{o_m}, \runB{o_1}, \ldots, \runB{o_m} \}$ where 
$e^i$, for $i \in \{1,2\}$, is obtained by replacing all registers $x$ with $x^i$ and all variables $v$ with  $v^i$ in expression $e$.

\mypara{Stuttering product circuit}
While the product circuit allows us to reason about pairs of executions, we need another ingredient to check microarchitectural contract satisfaction, as it refers to the \emph{filtered semantics} over a predicate $\phi$.
We cannot directly check the filtered semantics on the product circuit, as $\phi$ may be satisfied at different times. Instead, we modify the product circuit to synchronize the two executions based on $\phi$.
Given a circuit $C = \{ \passign{x_1}{e_1}, \ldots, \passign{x_n}{e_n} \} : 
\{ v_1 = e'_{1}, \ldots, v_k = e'_{k} \}  : {o_1, \ldots, o_m}$, we define its \emph{stuttering product circuit} over predicate $\phi$, denoted by $\stutteringProduct{C}{\phi}$, by replacing each assignment $\passign{\runA{x}}{\runA{e}}$ in the product circuit $C \times C$ with $\passign{\runA{x}}{\ite{\runA{\phi} \wedge \neg \runB{\phi}}{\runA{x}}{\runA{e}}}$ and, similarly, by replacing each $\passign{\runB{x}}{\runB{e}}$ in the product circuit $\product{C}{C}$ with $\passign{\runB{x}}{\ite{\runB{\phi} \wedge \neg \runA{\phi}}{\runB{x}}{\runB{e}}}$.
This transformation ensures that whenever $\phi$ holds in one execution but not the other, the execution where $\phi$ holds ``waits'' for the other one to catch up.
\begin{example}
    Consider the circuit $N = \{ \passign{i}{i+1} \} : \{\}: \{ i \}$.
    Forming the product yields $\product{N}{N}=\{ \passign{\runA{i}}{\runA{i}+1}, \passign{\runB{i}}{\runB{i}+1} \} : \{\}: \{ \runA{i},\runB{i} \}$.
    Let us define filter predicate $\phi=(i \; \mathbf{mod} \; 2 = 0)$. We get $\runA{\phi}=(\runA{i} \; \mathbf{mod} \; 2 = 0)$, and $\runB{\phi}=(\runB{i} \; \mathbf{mod} \; 2 = 0)$, and
    \begin{equation*}
        \stutteringProduct{N}{\phi} = \left \{
        \begin{array}[c]{@{}l@{}}
            \passign{\runA{i}}{\ite{\runA{\phi} \wedge \neg \runB{\phi}}{\runA{i}}{\runA{i}+1}}, \\
            \passign{\runB{i}}{\ite{\runB{\phi} \wedge \neg
            \runA{\phi}}{\runB{i}}{\runB{i}+1}} 
        \end{array}
        \right \} : \{\}: \{ \runA{i}, \runB{i} \} \ .
    \end{equation*}    
    Let us fix $\mu_I(\runA{i})=0$ and $\mu_I(\runB{i})=1$. 
    We only want to compare states where both $\runA{\phi}$ and $\runB{\phi}$ hold, \ie we want to compare the filtered semantics $\hwEvalInftyFilter{N}{\runA{\mu}}{\phi}$ and $\hwEvalInftyFilter{N}{\runB{\mu}}{\phi}$, where $\runA{\mu}(i)=0$ and $\runB{\mu}(i)=1$.
    In $\product{N}{N}$ the two executions are not synchronized and $\product{N}{N}, \mu_I \models \alwaysFuture (\runA{\phi} \leftrightarrow \neg \runB{\phi})$.
    In contrast, $\stutteringProduct{N}{\phi}$ synchronizes the two executions.
    As initially $\runA{\phi}$ holds but $\runB{\phi}$ does not, only $\runB{i}$ gets incremented and, afterwards, the two copies run in lockstep.
    We can now check properties of the filtered semantics, \eg that $\stutteringProduct{N}{\phi}, \mu_I \models \alwaysFuture (\runA{\phi} \wedge \runB{\phi} \to \runB{i}=\runA{i}+2)$ holds.
\end{example}

\mypara{Algorithm idea}
We now use the stuttering product circuit to verify that $\leakorder{\contract}{ \atk }{\uarch}{\uarchVars, \phi}$ holds.
This requires us to show that all executions whose filtered semantics produce the same contract observations always produce the same attacker observations (see \Cref{def:leakage-ordering}).
We start by adding an assumption to only consider executions of $\stutteringProduct{\uarch}{\phi}$ that are contract equivalent. We encode this via the formula $\Phi_{\mathit{ctr-equiv}} := (\runA{\phi} \wedge \runB{\phi} \to  \psi_{\mathit{equiv}}^\mathit{\contract})$, where $\psi_{\mathit{equiv}}^\mathit{M} := \bigwedge_{o \in M.O} \runA{o} = \runB{o}$ for a monitor $M$.
We then only consider executions that satisfy $\alwaysFuture \Phi_{\mathit{ctr-equiv}}$.
Next, our algorithm learns an inductive invariant $\mathit{LI}$ over the stuttering product circuit under our assumption. This invariant holds on all reachable states of the circuit.
Finally, our algorithm uses the invariant to prove that indeed all executions of the circuit are attacker equivalent.
For this we show that $\mathit{LI} \to \psi_{\mathit{equiv}}^\mathit{\atk}$ holds. Note that we prove this property over $\stutteringProduct{\uarch}{\phi}$, however the consequent of $\mathit{LI} \to \psi_{\mathit{equiv}}^\mathit{\atk}$ is stated over the unfiltered semantics.
To ensure that the stuttering semantics is equivalent to the regular one, we also prove $\mathit{LI} \to (\runA{\phi} \leftrightarrow \runB{\phi})$, \ie no stuttering occurs on contract equivalent traces.

\mypara{Algorithm description}
We implement this approach in \Cref{alg:verification}. 
It relies on the procedure \textsc{LearnInv}, which we use to learn invariants over the stuttering product circuit.
We first present \textsc{Verify} and later discuss \textsc{LearnInv}.

Function \textsc{Verify} is the entry point of our verification approach.
It takes as input a $\lang{}$ microarchitecture $\uarch$ (the processor under verification), a leakage monitor $\contract$ (capturing the allowed leaks), an attacker monitor $\atk$ (capturing what the attacker can observe), and a retirement predicate $\phi$.
To verify unbounded properties like $\leakorder{ \contract }{ \atk }{\uarch}{\uarchVars, \phi}$, the algorithm relies on inductive reasoning.
For this reason, \textsc{Verify} additionally take as input (1) a set of candidate invariants $\CandInvs$ over the stuttering circuit (which will be verified using \textsc{LearnInv}) as well as (2) a lookahead $b \in \Nat^{+}$. 
Concretely, \tool{} constructs the set of candidate invariants $\CandInvs$ directly from  $\uarch$,  $\atk$, and  $\phi$; see \Cref{sec:implementation} for more details.

In \cref{line:alg-verification:initial-formula}, we construct $\Phi_{\mathit{initial}}$ (over the stuttering circuit $\stutteringProduct{\uarch}{\phi}$) capturing the initial conditions for pairs of executions relevant to our check.
In $\Phi_{\mathit{initial}}$, $\runA{\psi_{\mathit{init}}^{\uarch}}$ and $\runB{\psi_{\mathit{init}}^{\uarch}}$ capture that the two executions start from valid initial states,  whereas $\psi_{\mathit{equiv}}^{\uarchVars}$ ensures that the two executions initially agree on all registers in $\uarchVars$, \ie{} $\psi_{\mathit{equiv}}^{\uarchVars} := \bigwedge_{x \in \uarchVars} \runA{x} = \runB{x}$.
In \cref{line:alg-verification:precondition-formula}, we construct $\Phi_{\mathit{ctr-equiv}} :=  (\runA{\phi} \wedge \runB{\phi} \to  \psi_{\mathit{equiv}}^\mathit{\contract})$ ensuring that contract observations are equivalent.
In \cref{line:alg-verification:learninv-call}, we call the \textsc{LearnInv} procedure to verify which of the candidate invariants in $\CandInvs$ are, indeed, invariants.
Hence, the learned invariants $\mathit{LI}$ hold for any two contract-indistinguishable executions, \ie $C \models (\Phi_{\mathit{initial}} \wedge \alwaysFuture \Phi_{\mathit{ctr-equiv}}) \to \alwaysFuture \bigwedge LI$ holds where $\bigwedge LI$ stands for $\bigwedge_{\phi \in LI} \phi$.
Finally, in \cref{line:alg-verification:security-check} we check whether the learned invariants are sufficient to ensure that (1) the attacker observations are the same and (2)  the predicate $\phi$ is always synchronized between the two executions.
If this is the case, \textsc{Verify} has successfully verified that $\leakorder{ \contract }{ \atk }{\uarch}{\uarchVars, \phi}$ holds; see \Cref{thm:soundness}.

The \textsc{LearnInv} procedure learns, using inductive verification, which of the candidate invariants are true invariants using an approach similar to the Houdini tool~\cite{flanagan2001houdini}.
\textsc{LearnInv} takes as input a circuit $C$, a formula capturing initial conditions $\Phi_{\mathit{initial}}$, a formula $\Phi_{\mathit{assumption}}$ that executions always need to satisfy, a bound~$b$, and a set of candidate invariants $CI$.
The procedure outputs the formulas in $CI$ that can be proved to be invariants, \ie for which $C \models (\Phi_{\mathit{initial}} \wedge \alwaysFuture \Phi_{assumption}) \to \alwaysFuture \bigwedge LI$ holds. %
Concretely, \textsc{LearnInv} consists of a base case (lines \ref{line:alg-verification:start-base}--\ref{line:alg-verification:end-base}) and an induction step (lines \ref{line:alg-verification:start-induction}--\ref{line:alg-verification:end-induction}).
Both parts follow a similar structure---they iteratively rule out invalid invariants based on counterexamples---and they differ only in the checked property:
$\Psi_{\mathit{base}}$ checks that for any state for which the initial conditions hold and for which the assumptions are satisfied for the next $b$ cycles, the invariants must also hold. %
In contrast, $\Psi_{\mathit{induction}}$ checks that for any state for which the invariants hold and for which the assumptions are satisfied for the next $b$ cycles, the invariants hold in the next cycle as well. %
Bound $b$ controls for how many cycles to unroll the assumption $\alwaysFuture \Phi_{assumption}$.
Unrolling the assumption is important for circuits where a difference in attacker observation occurs before a corresponding difference in contract observations.
This may happen, \eg if a leak occurs early in the pipeline and is later justified by a difference in contract observations at retirement. 
It therefore often suffices to bound $b$ by the processor's pipeline depth.

\begin{algorithm}[t]
    \caption{\tool{} verification approach}\label{alg:verification}
\begin{algorithmic}[1]
	\Require \!Microarchitecture $\uarch$, leakage monitor $\contract$, attacker $\atk$, retirement predicate $\phi$, lookahead $b$, candidate invariants $CI$\looseness=-1
	\smallskip
    \Procedure{Verify}{$\uarch , \contract, \atk, \phi, b, CI$}\label{line:alg-verification:begin-verify}
        \State{$\Phi_{\mathit{initial}} := \runA{\psi_{\mathit{init}}^{\uarch}} \wedge \runB{\psi_{\mathit{init}}^{\uarch}} \wedge \psi_{\mathit{equiv}}^{\uarchVars}$}\label{line:alg-verification:initial-formula}
        \State{$\Phi_{\mathit{ctr-equiv}} :=  (\runA{\phi} \wedge \runB{\phi} \to  \psi_{\mathit{equiv}}^\mathit{\contract})$}\label{line:alg-verification:precondition-formula}
        \State{$LI := \textsc{LearnInv}(\stutteringProduct{\uarch}{\phi}, \Phi_{\mathit{initial}}, \Phi_{\mathit{ctr-equiv}}, b, CI)$}\label{line:alg-verification:learninv-call}
        \State{\Return{$\stutteringProduct{\uarch}{\phi} \models \bigwedge LI \to \psi_{\mathit{equiv}}^\mathit{\atk} \wedge (\runA{\phi} \leftrightarrow \runB{\phi})$}}\label{line:alg-verification:security-check} 
    \EndProcedure\label{line:alg-verification:end-verify}
    \smallskip
    \Procedure{LearnInv}{$C, \Phi_{\mathit{initial}}, \Phi_{\mathit{assumption}}, b, CI$}\label{line:alg-verification:begin-learninv}
        \While{$\top$} \Comment{base case}\label{line:alg-verification:start-base}
            \State{$\Psi_{\mathit{base}} := (\Phi_{\mathit{initial}} \wedge \boundedFuture{b} \Phi_{\mathit{assumption}}) \to  \bigwedge CI $} %
            \If{$C \models \Psi_{\mathit{base}} $}
                \State{\textbf{break}}
            \Else
                \State{Let $\val$ be the counterexample}
                \State{$CI:= \{ \phi \in CI \mid C, \val \models \phi \}$}
            \EndIf
        \EndWhile\label{line:alg-verification:end-base}

        \While{$\top$} \Comment{ind. step}\label{line:alg-verification:start-induction}
        \State{$\Psi_{\mathit{inductive}} := (\bigwedge CI \wedge \boundedFuture{b} \Phi_{\mathit{assumption}})  \to \future \bigwedge CI $}\label{line:alg-verification:ind-invariant}
        \If{$ C \models \Psi_{\mathit{inductive}}$}
            \State{\Return{$CI$}}
        \Else
            \State{Let $\val$ be the counterexample}
            \State{$CI:= \{ \phi \in CI \mid C, \val \models \phi \}$}
        \EndIf
    \EndWhile\label{line:alg-verification:end-invariants}\label{line:alg-verification:end-induction}
    \EndProcedure\label{line:alg-verification:end-learninv}

\end{algorithmic}
\end{algorithm}

\paragraph{Soundness}
\Cref{thm:soundness}  states that whenever \Cref{alg:verification}  returns $\top$, then microarchitectural contract satisfaction holds. %

\begin{theorem}\label{thm:soundness}
$\textsc{Verify}(\uarch,\contract,\atk,\phi, b , \RelInvs) \!\Rightarrow\! \leakorder{ \contract\!}{\!\!\!\atk }{\uarch}{\uarchVars, \phi}$.\looseness=-1 %
\end{theorem}
\begin{example}
    Consider again the implementation $\SimpIMP$ from \Cref{ex:muarch}. We want to
    verify that $\leakorder{ \SimpLM }{ \SimpATK }{\SimpIMP}{\{ \mathit{st},
    \mathit{res}, \mathit{ret} \}, \phi}$ holds. We start by building the
    stuttering product circuit $\stutteringProduct{\SimpIMP}{\phi}$ with respect to
    retirement predicate $\phi=(ret=1)$.
    We can assume that the two executions produce the same contract observations,
    whenever both executions retire. We capture this assumption in formula
    $\Phi_{\mathit{ctr-equiv}} := (\runA{ret}=1 \wedge \runB{ret}=1 \to
    \runA{m}[\runA{pc}] = \runB{m}[\runB{pc}])$, which we assume to hold
    throughout the execution. Next, we want to learn an inductive invariant over
    $\stutteringProduct{\SimpIMP}{\phi}$ under assumption $\alwaysFuture \ \Phi_{\mathit{ctr-equiv}}$.
    We pick the following set of candidate invariants.
    \begin{equation*}
        \CandInvs = \; \left \{ 
        \begin{array}[c]{@{}l@{}}
            \runA{pc} = \runB{pc}, \runA{st} = \runB{st}, \runA{res} = \runB{res}, \runA{ret} = \runB{ret} \\[\jot] 
            \runA{st}=0 \to \runA{ret}=1, \runA{st}=1 \to \runA{ret}=1 \\[\jot]  
        \end{array}
        \right \}       
    \end{equation*}
    Procedure \textsc{LearnInv} starts by checking the invariant candidates on
    the initial state. We set bound $b$ to 1. Since in all valid initial states
    $\mu$, we have $\mu(pc)=1$, $\mu(st)=0$, and all microarchitectural
    variables are assumed to be equal via $\Phi_{\mathit{initial}}$ we retain
    all candidate invariants.
    Next, \textsc{LearnInv} checks whether the candidate invariants are preserved under
    transitions. 
    That is, if we assume the invariant holds and take a transition step, the invariant must still hold.
    Since our invariant does not require memory $m$ to be equal in both executions,
    taking the else branch in line 3 of $\SimpIMP$ (see \Cref{ex:muarch}) produces a
    counterexample where $\runA{res} \neq \runB{res}$ and we remove the
    corresponding invariant.
    Similarly, taking the else branch in line 3 produces a state where
    $\runA{st}=1$ and $\runA{ret}=0$ and \textsc{LearnInv} removes the invariant as well. 
    The remaining candidate invariants are preserved under transitions and the procedure returns. 
    This leaves us with the following set of learned invariants.
    \begin{equation*}
        LI = \; \left \{ 
        \begin{array}[c]{@{}l@{}}
            \runA{pc} = \runB{pc}, \runA{st} = \runB{st}, \runA{ret} = \runB{ret} \\[\jot] 
            \runA{st}=0 \to \runA{ret}=1 \\[\jot]  
        \end{array}
        \right \}       
    \end{equation*}
    Finally, procedure \textsc{Verify} checks whether the conjunction of the
    learned invariants implies that attacker observations and retirement are the
    same in both executions.
    For our example, this means checking that the following implication holds.
    \begin{equation*}
        \left ( 
        \begin{array}[c]{@{}c@{}}
            \runA{pc} = \runB{pc} \wedge \runA{st} = \runB{st} \wedge \\[\jot] 
            \runA{ret} = \runB{ret} \wedge \runA{st}=0 \to \runA{ret}=1  
        \end{array}
        \right ) \to \left (
        \begin{array}[c]{@{}l@{}}
            \runA{pc} = \runB{pc} \wedge \\[\jot] (\runA{ret} = 1) \leftrightarrow (\runB{ret} = 1)    
        \end{array}
        \right )
    \end{equation*}
    As the implication is valid, we have proved microarchitectural contract satisfaction.
   \end{example}

\section{Implementation}\label{sec:implementation}

In this section, we present the \tool{} verification tool, which implements the verification approach from~\Cref{sec:verification:algorithm} for Verilog.
\tool{} uses the Yosys Open Synthesis Suite~\cite{yosys} for processing Verilog circuits, the Icarus Verilog simulator~\cite{iverilog} for simulating counterexamples, and the Yices SMT solver~\cite{yices} for verification.
\tool{} is open source and available at~\cite{leave-artifact} together with the benchmarks and scripts for reproducing the experiments from~\Cref{sec:evaluation}.

\paragraph{Inputs} \tool{} takes as input
\begin{inparaenum}
    \item the processor under verification (PUV) $\uarch$ implemented in Verilog,
    \item a leakage monitor formalized as Verilog expressions over $\uarch$'s architectural 
    state, 
    \item an attacker expressed as Verilog expressions over $\uarch$, 
    \item a retirement predicate $\phi$ expressed as a Boolean condition over $\uarch$, and
    \item a lookahead $b \in \Nat^+$.\footnote{As a rule of thumb, a sufficient choice for $b$ is the maximum number of cycles needed for an instruction to traverse the pipeline (from fetch to retire).}
\end{inparaenum}
Users can provide candidate relational invariants as expressions $e$ over $\uarch$ and \tool{} will construct the candidate invariant $\runA{e} = \runB{e}$.
Users can also provide additional invariants over individual executions of $\uarch$ to help ruling out spurious counterexamples.\footnote{\tool{} only verifies the relational invariants, which concern security. Invariants over $\uarch$, which concern functional correctness, are assumed and not checked by the tool.\looseness=-1}

\paragraph{Workflow}
\tool{} works in two steps that follows  \Cref{alg:verification}.

First, \tool{} determines the greatest subset of the provided candidate relational invariants that is inductive.
For this, \tool{} implements the \textsc{LearnInv} function from \Cref{alg:verification} (described below).
In addition to the user provided candidate invariants,  the set of candidate invariants for \textsc{LearnInv} contains:
\begin{inparaenum}
    \item \emph{all} relational formulas of the form $\runA{x} = \runB{x}$ where $x$ is a register or wire in $\uarch$,
    \item formulas of the form $\runA{e_\atk} = \runB{e_\atk}$ for all expressions $e_\atk$ in the provided attacker, and
    \item the invariant $\runA{\phi} \leftrightarrow \runB{\phi}$ indicating that the retirement predicate is always synchronized between the two executions.
\end{inparaenum} 

Next, \tool{} analyzes the learned invariants to determine if they are sufficient to prove security with respect to the given attacker.
For this, \tool{} checks if the invariants associated with the attacker and with the retirement predicate are part of the set of learned invariants, which is sufficient to ensure the satisfaction of the check at line 5 in \Cref{alg:verification}. %

\paragraph{Implementation of \textsc{LearnInv}}
\tool{}'s implementation of \textsc{LearnInv} follows \Cref{alg:verification}:
\begin{inparaenum}
    \item It constructs the stuttering product circuit by combining two copies of the PUV and using the provided retired predicate $\phi$ to synchronize the two executions (as described in~\Cref{sec:verification:algorithm}).
    \item Then, it inlines the property to be verified (\ie{} $\Psi_{\mathit{base}}$ and $\Psi_{\mathit{inductive}}$ from \Cref{alg:verification}) as \texttt{assume} and \texttt{assert} Verilog statements in the product circuit.
    \item Next, it checks whether the property holds.
    \item Whenever a property is not satisfied, \tool{} analyzes the counterexample to determine which candidate relational invariants are violated (lines 12-13 and 19-20 in \Cref{alg:verification}) 
\end{inparaenum}

For (1) and (2), we implemented dedicated Yosys passes that construct the stuttering product circuit and inline candidate relational invariants.
For (3), \tool{} uses Yosys to encode the product circuit and the verification queries into SMT logical formulas and the Yosys-BMC~\cite{yosys} backend to verify the property with the Yices SMT solver (using the lookahead $b$ as verification bound). %
For (4), when verification fails, Yosys-BMC translates the SMT counterexample into a Verilog testbench.
\tool{} instruments the testbench to monitor the value of all candidate invariants, simulates the testbench using Icarus Verilog, and discards the violated invariants.

\newcommand{\PipelineMul}{\textbf{RE}}
\newcommand{\DarkRiscv}{\textbf{DarkRISCV}}
\newcommand{\DarkRiscvTwoStages}{\DarkRiscv{}\textbf{-2}}
\newcommand{\DarkRiscvThreeStages}{\DarkRiscv{}\textbf{-3}}
\newcommand{\Sodor}{\textbf{Sodor}}
\newcommand{\SodorOneStage}{\Sodor\textbf{-1}}
\newcommand{\SodorTwoStages}{\Sodor\textbf{-2}}
\newcommand{\Ibex}{\textbf{Ibex}}
\newcommand{\IbexSmall}{\Ibex{}\textbf{-small}}
\newcommand{\IbexMul}{\Ibex{}\textbf{-mult-div}}
\newcommand{\IbexCache}{\Ibex{}\textbf{-cache}}
\newcommand{\PcContract}{\textbf{I}}
\newcommand{\BranchContract}{\textbf{B}}
\newcommand{\AlignedContract}{\textbf{A}}
\newcommand{\MemContract}{\textbf{M}}
\newcommand{\MulContract}{$\textbf{O}_{\textbf{m}}$}
\newcommand{\DivContract}{$\textbf{O}_{\textbf{d}}$}
\newcommand{\RetireAtk}{\textbf{R}}
\newcommand{\InstMemoryAtk}{\textbf{I}}
\newcommand{\DataMemoryAtk}{\textbf{D}}
\newcommand{\FourWayTool}{\textsc{4way-}\tool{}}

\section{Evaluation}\label{sec:evaluation}

This section reports on our use of \tool{} to verify the security of three open-source RISC-V processors.
We start by introducing our methodology (\Cref{sec:evaluation:setup}): the processors we analyze, the leakage contracts and attacker we consider, and the experimental setup.
In our experimental evaluation (\Cref{sec:evaluation:verification-results}), we address the following three research questions:
\begin{inparaenum}
    \item[\textbf{Q1}:] Can \tool{} be used to reason about the security of open-source RISC-V processors?
    \item[\textbf{Q2}:] What is the impact of varying the lookahead $b$ on verification time?
    \item[\textbf{Q3}:] What is the impact of decoupling security and functional correctness on verification? %
\end{inparaenum}\looseness=-1

\subsection{Methodology}\label{sec:evaluation:benchmarks}

\paragraph{Benchmarks}
We consider the following benchmarks.
\begin{asparaitem}
\item \PipelineMul: The simple processor from \Cref{sec:overview}. 
The \var{log_time_mul} module is implemented using shift operations (logarithmic in the number of set bits of the multiplier), inspired by one of Ibex's multipliers~\cite{ibex-mult}. %

\item \DarkRiscv: A RISC-V processor implementing most of the RISC-V RV32E and RV32I instruction set~\cite{darkriscv}.
The processor is in-order and single-issue, and we analyzed its 2-stage (\DarkRiscvTwoStages) and 3-stage (\DarkRiscvThreeStages) versions. %

\item \Sodor{}: An educational RISC-V processor~\cite{sodor}.
We analyzed the 2-stage version of Sodor implementing the RV32I instruction set.\looseness=-1

\item \Ibex{}: An open-source, production-quality 32-bit RISC-V CPU core~\cite{ibex}.\footnote{Ibex is written in SystemVerilog. To analyze it with \tool{}, we first translate it into plain Verilog using scripts from Ibex's developers.} 
We target Ibex in its default configuration (called ``small''~\cite{ibex}), 
which underwent functional correctness verification.
The processor has two stages and supports the RV32IMC instruction set.
In our experiments, we consider three variants of Ibex:
(1) \IbexSmall{} is the default ``small'' configuration with constant-time multiplication (three cycles) and without caches,
(2) \IbexCache{} is the \IbexSmall{} version extended with a simple (single-line) cache, and
(3) \IbexMul{} employs a non-constant-time multiplication unit whose execution time depends on the operands~\cite{ibex-mult}.

\end{asparaitem}

\smallskip
For the RISC-V processors in our experiments (\ie all variants of \DarkRiscv, \Sodor, \Ibex), we make the following assumptions during verification: 
\begin{inparaenum}[(1)]
\item debug mode is disabled,
\item all fetched instructions are legal and not compressed,
\item no exceptions or interrupts are raised during execution, and
\item only unprivileged instructions are executed.
\end{inparaenum}
Additionally, for 
 \IbexCache, we assume that memory operations are aligned at word boundaries due to limitations of our simple cache implementation. 
Finally, for all processors we manually specify a retirement predicate indicating when instructions retire.\looseness=-1

\paragraph{Leakage contracts}
We consider leakage contracts constructed by composing the following building blocks:
\begin{asparaitem}
\item \PcContract{}: This contract exposes the architectural program counter and the corresponding instruction retrieved from  memory.
\item \BranchContract{}: This contract exposes the architectural outcome of (direct and indirect) branch instructions. 
That is, for conditional branches, the contract exposes the architectural value of the condition.
\item \MemContract{}: This contract exposes the addresses accessed by load and store memory instructions.
\item \AlignedContract{}: This contract exposes whether load and store memory instructions are aligned.
\item \MulContract{}: This contract exposes the operands of \var{mul} and \var{imul} multiplication instructions. %
\item \DivContract{}: This contract exposes whether the divisor in \var{div} (division) and \var{rem} (remainder) instructions is $0$. %
\end{asparaitem}
In the following, we write \textbf{A}+\textbf{B} to denote the composition of contracts \textbf{A} and \textbf{B}.
For instance, \PcContract{}+\BranchContract{}+\MemContract{} is the contract that exposes everything exposed by \PcContract, \BranchContract, and \MemContract{}. This contract corresponds to the standard constant-time model~\cite{almeida2016verifying}.
We order contracts by the amount of information they leak, where stronger contracts leak less. For example, \PcContract{} is stronger than \PcContract{}+\BranchContract{} as it exposes less information.
For each processor from \Cref{sec:evaluation:benchmarks}, we implemented all the above mentioned contracts and their combinations as leakage monitors over the processor's architectural state.

\paragraph{Attacker}
For all processors from \Cref{sec:evaluation:benchmarks}, we implemented an attacker monitor that observes when instructions retire by exposing the value of the retirement predicate at each cycle.

\paragraph{Additional candidate invariants}
For $\DarkRiscv{}$, $\Sodor$, and $\Ibex$, we manually specified candidate relational invariants capturing that ``if instructions enter a pipeline stage in both executions, then the instructions are the same in both executions''. %
Moreover, for $\IbexCache{}$, we also added a candidate invariant capturing that ``if both executions are executing a load instruction, then the signals detecting a cache hit are the same.''
All these invariants can be formalized as formulas of the form $\runA{e} = \runB{e} \rightarrow \runA{e'} = \runB{e'}$.
These are not part of the invariants automatically generated by \tool{}, which are of the simpler form $\runA{e} = \runB{e}$.

\paragraph{Experimental setup}\label{sec:evaluation:setup}
All our experiments are run on a Ubuntu 20.04 virtual machine with 8 CPU cores and 32 GB of RAM running on Linux KVM on a server with 4 Xeon Gold 6154 CPUs and 512~GB of DDR4 RAM.
We configured \tool{} to run with Yosys version $0.24+10$, Icarus Verilog version $12.0$, and Yices version $2.6.4$.

\subsection{Experimental results \label{sec:evaluation:verification-results}}
\paragraph{Q1: Reasoning about open-source processors}
To evaluate whether \tool{} can verify the security guarantees of open-source processors, we use it to prove microarchitectural contract satisfaction  against an attacker $\atk$ that observes when instructions are retired. %
For each processor and leakage monitor from~\Cref{sec:evaluation:benchmarks}, we use \tool{} to check whether the attacker monitor leaks less than the leakage monitor with respect to the processor and its retirement predicate (indicating whenever instructions retire).

\Cref{fig:verification:results} reports (1) the strongest contract that could be verified against $\atk$, (2) the time needed for the verification of the satisfaction of the strongest contract%
, (3) the total number of iterations taken by \textsc{LearnInv} for the base and induction steps (i.e., the number of issued SMT queries), and (4) the minimum lookahead $b$ for which verification succeeded.
We highlight the following findings:
\begin{asparaitem}
\item For the \PipelineMul{} processor from \Cref{sec:overview}, \tool{} successfully verified contract satisfaction against the contract \MulContract{} exposing the multiplication's operand in 1.5 minutes with a lookahead of $33$. 
Such a lookahead is needed to ensure that in-flight multiplications are retired and the corresponding contract observation is produced. %

\item For \DarkRiscv, \tool{} proves contract satisfaction against the \PcContract{} contract, which exposes the current program counter and the loaded instruction, in 7 minutes for the two-stage version \DarkRiscvTwoStages{} and in around 11 minutes for the (more complex) three-stage version \DarkRiscvThreeStages{}.

\item Differently from \DarkRiscv{},  \SodorTwoStages{} only satisfies the weaker \PcContract{}+\BranchContract{} contract, which additionally exposes the outcome of branch instructions. %
This arises from the processor employing a simple form of branch prediction, which predicts that the branch is always not taken. This results in a timing leak because mispredictions trigger a pipeline flush.
Consider the following instruction (returned by \tool{} as a counterexample when trying to prove satisfaction against \PcContract{}) $i \eqdef \var{beq} \ t_1 \ t_2 \ \mathit{pc}+4$ at address $\mathit{pc}$, which conditionally jumps to $\mathit{pc}+4$ if registers $t_1$ and $t_2$ have the same value. 
The next instruction will \emph{always} be the one at address \var{pc+4} (so, executions will be equivalent under contract \PcContract{}).
However, executing $i$ on \SodorTwoStages{} takes a different number of cycles depending on whether $t_1$ and $t_2$ are equal. 

\item For \IbexSmall{}, \tool{} can only prove security against the \PcContract{}+\BranchContract+\DivContract{}+\AlignedContract{} contract, which additionally exposes (a) whether the divisor in division and remainder instructions is 0 and (b) whether memory accesses are aligned.
The \DivContract{} is needed to capture that division and remainder operations take 1 cycle when the divisor is 0 or 37 cycles otherwise.
Moreover, the \AlignedContract{} contract is needed to capture that \Ibex{} handles memory accesses that are not aligned on word boundaries by performing two separate word-aligned memory accesses.
Note that the difference in complexity between \SodorTwoStages{} (a simple educational processor) and \IbexSmall{} (a production-quality processor) is reflected in the difference in the time taken by a single \textsc{LearnInv} iteration (1.1 versus 16 minutes on average) and by the larger lookahead (1 vs 38).

\item For \IbexCache{}, \tool{} can only prove security against the \PcContract{}+\BranchContract{}+\DivContract{}+\MemContract{} contract.
Differently from \IbexSmall{}, which is secure against the \PcContract{}+\BranchContract{}+\DivContract{}+\AlignedContract{} contract, \IbexCache{} needs the \MemContract{} contract that exposes the accessed memory addresses (rather than the alignment bit).
This reflects the effects of our single-line cache which requires 3 cycles for hits and 4 cycles for misses.
\item For \IbexMul{}, \tool{} can only prove security against the \PcContract{}+\BranchContract{}+\DivContract{}+\MulContract{}+\AlignedContract{} contract, which also exposes the operands of multiplication instructions (\MulContract).
This captures the effects of the non-constant-time multiplier used in \IbexMul{}, whose execution time is proportional to the logarithm of the multiplication operands.

\end{asparaitem}

\newcolumntype{R}[1]{>{\raggedleft\let\newline\\\arraybackslash\hspace{0pt}}m{#1}}

\begin{table}

    \caption{Verification results for our benchmarks. For each processor, the table indicates the strongest satisfied contract (i.e., the one exposing the least amount of information) against an attacker observing when instructions retire.%
}\label{fig:verification:results}
    \centering
    \scalebox{.9}{
\begin{tabular}{l c c c c}
\toprule
    \multirow{2}{*}{\textit{Processor}} & \textit{Strongest} & \textit{Verification time} & \textsc{LearnInv} & \multirow{2}{*}{$b$} \\
     & \textit{contract} & \textit{(in minutes)}  &  \textit{iterations} & \\
    \midrule
    \PipelineMul{}    & \MulContract{} & 1.5 &  10  & 33\\
    \DarkRiscvTwoStages{} & \PcContract{} & 7.2  &  52  & 2 \\
    \DarkRiscvThreeStages{} & \PcContract{} & 11.1  & 83   &  2\\ 
    \SodorTwoStages{} & \PcContract{}+\BranchContract{} & 97.8  &  85  & 1 \\
    \IbexSmall{} & \PcContract{}+\BranchContract{}+\DivContract{}+\AlignedContract{}  & 1479.4  &  90  & 38 \\
    \IbexCache{} & \PcContract{}+\BranchContract{}+\DivContract{}+\MemContract{} & 1396.7  & 67   & 38 \\
    \IbexMul{} & \PcContract{}+\BranchContract{}+\DivContract{}+\MulContract{}+\AlignedContract{}  & 1291.9 & 75   & 38 \\
    \bottomrule
\end{tabular}
    }
\end{table}

\paragraph{\textbf{Q2}: Impact of lookahead}\label{sec:evaluation:lookahead}
\tool{}'s verification queries are parametric in the lookahead $b$. 
A larger lookahead corresponds to stronger assumptions and may thus enable learning stronger invariants. %
This, however, comes at the cost of more complex queries to the SMT solver, which increases solving time.
To understand the impact of increasing $b$, we use \tool{} to analyze the \Sodor{} processor against the \PcContract{}+\BranchContract{} contract for different values of $b \in \{ 1, 2, 3, 5, 10, 20\}$.

\Cref{fig:verification:loakahead} reports the total verification time, the total number of iterations taken by the \textsc{LearnInv} sub-procedure for the base and induction steps (i.e., the number of issued SMT queries), the time per iteration, and the number of invariants learned.
Our results indicate that increasing the lookahead $b$ results in slower iterations of \textsc{LearnInv} and in more invariants.
For instance, increasing the bound from $1$ to $20$ results in increasing the iteration time from 1.15 to 8.29 minutes.
The total number of \textsc{LearnInv} iterations (and, thus, the total verification time), however, varies depending on which counterexamples the SMT solver returns.

\paragraph{\textbf{Q3}: Impact of decoupling}\label{sec:evaluation:reduction}
To understand the impact of checking microarchitectural contract satisfaction using our decoupling theorem versus checking contract satisfaction directly using an architectural model $\arch$ (c.f.~\Cref{def:contract-satisfaction}), we modified \tool{} to directly prove contract satisfaction according to \Cref{def:contract-satisfaction}. %
For this, we (1) replace the construction of the stuttering circuit $\stutteringProduct{\uarch}{\phi}$ with the product circuit $\product{\arch}{ \product{\arch}{\product{\uarch}{\uarch}} }$ and (2) modify the construction of the $\Psi_{\mathit{initial}}$ and $\Psi_{\mathit{contract}}$ formulae, whereas the rest (e.g., the \textsc{LearnInv} procedure) is the same.
We refer to this modified version of \tool{} as \FourWayTool{} \onlyShortVersion{(see \cite[Appendix D]{techReport})}\onlyTechReport{(see \Cref{alg:verification-4way} in appendix)}. %
Note that \FourWayTool{} and \tool{} prove different properties which, as stated in \Cref{thm:soundness}, are equivalent only for ISA-compliant designs.

We analyzed the \Sodor{} processor against the \PcContract{}+\BranchContract{} contract using both \tool{} and \FourWayTool{}.
We focused our analysis on \Sodor{} because it comes with a Verilog ISA model (i.e.,  1-stage \Sodor{}).

In our experiments, when using a lookahead of $b = 2$, \tool{} successfully proved that the \PcContract{}+\BranchContract{} contract is satisfied in $97.8$ minutes.
In contrast, \FourWayTool{} tool proved contract satisfaction in $33.5$ hours.
This illustrates that, even for the simple 2-stage \Sodor{} processor, directly proving \Cref{def:contract-satisfaction} is impractical.
It also confirms that our decoupling theorem is instrumental in enabling practical automated proofs of contract satisfaction for realistic hardware.

\begin{table}[t]
    \caption{Verification time for \Sodor{} against the \PcContract{}+\BranchContract{} contract for different lookaheads $b$.}\label{fig:verification:loakahead}
	\centering
    \scalebox{.92}{
        \begin{tabular}{r c c c c}
            \toprule
            \multirow{2}{*}{$b$} & \textit{Verification time} & \textit{\textsc{LearnInv}} & \textit{Time per iteration} & \textit{Number of}\\
             & \textit{ (in minutes)} & \textit{iterations} & \textit{(in minutes)} & \textit{invariants}\\
            \midrule
            1    & 97.8 & 85  & 1.15 & 1732\\
            2    & 81.9 & 63  &  1.30 & 1770\\
            3   & 123.5 & 81  &  1.52 & 1776\\
            5   &  102.9 & 59 &  1.74 & 1776\\
            10   & 174.4 & 63 &  2.77 & 1776\\
            20   & 497.6 & 60 &  8.29 & 1776\\
            \bottomrule
        \end{tabular}
    }
\end{table}

\section{Discussion}\label{sec:discussion}

\paragraph{Limitations}
Our formalization of leakage contracts and our notion of ISA compliance  impact both the decoupling theorem (\Cref{thm:contract-sat-preconditions}) as well as the microarchitectures and contracts supported by \tool{}.
In terms of microarchitectures, our notion of ISA compliance (\Cref{def:isa-compliance}) only applies to single-issue processors.
Supporting multi-issue processors requires an ISA compliance notion that accounts for retiring multiple instructions per cycle. 
In terms of leakage contracts, \tool{} targets  \emph{sequential} leakage contracts that only refer to ``architectural'' instructions.
We leave the support for leakage contracts that refer to transient instructions, like the speculative contracts from~\cite{contracts2021}, as future work.\looseness=-1

We also remark that (1) \tool{} currently lacks support for inputs, and (2) our formalization of attackers as combinatorial monitoring circuits limits \tool{} to reason about \emph{passive attackers} that can only observe (part of) a processor's microarchitecture during execution.
We leave corresponding extensions to future work.

\paragraph{Lookahead}
\tool{}'s verification approach is parametric in a lookahead $b \in \mathbb{N}^+$ which determines for how many cycles the contract-equivalence assumption needs to be unrolled in the verification queries issued by the \textsc{LearnInv} function in \Cref{alg:verification}.
The lookahead $b$ is used to expose contract observations (produced at retirement) for instructions that are in-flight.
In particular, it allows accounting
for microarchitectural differences at cycle $i$ that are later declassified by a contract observation produced at cycle (at most) $i+b$.
We remark that the choice of $b$ does not affect the soundness of \tool{} (cf.~\Cref{thm:soundness}), but it may affect the success of verification. 
For instance, verifying the satisfaction of the contract \excontract{} from \Cref{sec:overview} for the processor from \Cref{fig:running} requires $b = 33$ (see \Cref{fig:verification:results}), \ie verification fails for smaller bounds.
In our experiments, setting $b$ to the maximum number of cycles needed for instructions to traverse the pipeline (from fetch to retire) was always sufficient whenever contract satisfaction holds.

\paragraph{Leakage contracts and secure programming}
Leakage contracts may serve as a foundation for secure programming.
As shown by Guarnieri et al.~\cite{contracts2021}, ensuring at program level that secret data do not influence leakage contract traces is sufficient to ensure the absence of leaks at microarchitectural level for processors that satisfy the contract.
Thus, \tool{}'s verification results have direct implications for programmers.
As an example, for each processor in our evaluation (\Cref{sec:evaluation}), the strongest contract verified by \tool{}, reported in \Cref{fig:verification:results}, indicates which parts of a computation should not involve secrets to ensure leakage freedom.
For instance, secure programming for the contract \PcContract{}+\BranchContract{}+\DivContract{}+\MulContract{}+\AlignedContract{}, satisfied by \IbexMul, requires ensuring that secrets do not influence (i) the program's control-flow (\PcContract{}+\BranchContract), (ii) whether the divisor in \texttt{div} and \texttt{rem} instructions is $0$ (\DivContract{}),  (iii) the operands of \texttt{mul} and \texttt{imul} instructions (\MulContract), and (iv) the alignment of memory accesses (\AlignedContract{}).
\section{Related work}\label{sec:rw}

\paragraph{Hardware verification for security}
UPEC~\cite{fadiheh2022exhaustive} is an approach for detecting confidentiality violations in RTL circuits.
Similarly to \Cref{def:leakage-ordering}, the UPEC property is defined as a non-interference-style property over pairs of microarchitectural executions.
However, the security property verified in~\cite{fadiheh2022exhaustive} is fixed; it specifically focuses on microarchitectural leaks due to transitive execution; and it is not directly based on an ISA-level specification, \ie{} it does not correspond to a leakage contract in a straightforward manner.
In contrast, our approach directly supports leakage contracts defined at ISA-level.\looseness=-1

Bloem et al.~\cite{powercontracts} propose an approach for verifying power leakage models (formalized on top of the Sail domain specific language~\cite{conf/popl19/armstrong}) for RTL circuits, which  differs from \tool{} in two key ways:
\begin{inparaenum}
\item They target power side channels, whereas \tool{} focuses on software-visible microarchitectural leaks. 
This is reflected in different notions of contract satisfaction: the one from~\cite{powercontracts} is probabilistic and related to threshold non-interference, whereas ours is related to standard non-interference.
\item 
Their verification approach needs a user-provided simulation mapping that ``specifies for all registers in the hardware [..] a location in the contract modeling the hardware location''~\cite[\S3.4]{powercontracts} where a location is a register or an input.
Defining such a mapping can be challenging for complex processors, \eg{} registers of stateful microarchitectural components (like caches or predictors) may depend on multiple instructions.
\tool{} does not need such a mapping for checking contract satisfaction; it only needs (automatically synthesized or manually provided) candidate relational invariants over the microarchitecture.

\end{inparaenum}

{Iodine}~\cite{gleissenthall2019iodine} and {Xenon}~\cite{v2021solver} check if the execution time of an RTL circuit is \emph{input independent} given a partitioning of the circuit's inputs into secret and public.
This partitioning is too coarse to support leakage contracts, where the notion of what is ``secret'' depends on the executed instructions. %
Finally, secure Hardware Description Languages~\cite{zhang2015hardware,deng2019secchisel} aim at building secure processors by construction.
They require partitioning RTL registers and inputs into secret and public, which is too coarse-grained for leakage contracts.\looseness=-1

Knox~\cite{athalye2022verifying} is a verification approach for  hardware security modules (HSMs) that targets an HSM's  hardware and software components.
While leakage contracts capture a processor's security guarantees at ISA level, Knox focuses on ensuring that all components of an HSM are \emph{both}  functionally correct and leakage free.
Differently from \tool{}, Knox relies on a combination of annotations and interactive proofs.\looseness=-1

\paragraph{Hardware verification for functional correctness}
A multitude of approaches for verifying functional correctness of processors have been proposed~\cite{reid2016end,Huang19, Zeng21,burch1994automatic, khune2010automated,patankar1999formal,jhala12001microarchitecture}.
Some of these approaches adopt a notion of ISA compliance similar to \Cref{def:isa-compliance}.
For instance, Reid et al.~\cite{reid2016end} illustrate a verification approach (used internally at ARM) for checking compliance between a microarchitecture and a reference architectural model, where the notion of ISA compliance requires that all changes to the architectural state are reflected by a ``step'' of the reference model (similarly to \Cref{def:isa-compliance}).

The Instruction-Level Abstraction (ILA) project~\cite{Huang19, zhang2020synthesizing, Zeng21} aims to specify and verify instruction-level models of processors and accelerators. 
They present techniques for (1) checking whether an RTL implementation correctly implements an ILA model, (2) determining which parts of a processor's state are architectural~\cite{Zeng21}, and (3) deriving processor invariants~\cite{zhang2020synthesizing}.
Some of these techniques can help in \tool{}'s verification.
For instance,~\cite{Zeng21} can help in identifying the $\archVars$ and $\uarchVars$ sets, whereas~\cite{zhang2020synthesizing} can complement \tool{}'s invariant learning approach.

Finally, fuzzing approaches~\cite{280028,canakci2022processorfuzz} can detect violations of ISA compliance, but they cannot {prove} functional correctness.

\paragraph{Detecting leaks through testing}
Revizor~\cite{oleksenko2022revizor,oleksenko2023hide} and Scam-V~\cite{Nemati2020a,buiras2021micro} search for contract violations (\ie{} they find counterexamples to \Cref{def:contract-satisfaction}) for black-box CPUs.
However, they require physical access to a CPU and can be applied only post-silicon.
Other approaches~\cite{Weber2021,Gras2020,Moghimi2020a} instead detect leaks by analyzing hardware measurements without the help of a formal leakage model but, again, apply only post-silicon.
Finally, SpecDoctor~\cite{hur2022specdoctor} and SigFuzz~\cite{rajapakshasigfuzz} can test for leaks on RTL designs and they are applicable in the pre-silicon phase.
Differently from \tool{}, all these approaches \emph{cannot} prove the absence of leaks.

\paragraph{Formal leakage models}
Researchers have proposed many formal models for studying microarchitectural security at program level, ranging from simple models associated with ``constant-time programming''~\cite{almeida2016verifying,molnar2005program} to more complex ones capturing leaks associated with speculatively executed instructions~\cite{spectector2020,patrignani2021exorcising,fabian202automatic,pitchfork, GuancialeBD20, blade}.
Most of these models focus at the software level and have no formal connection with leaks in hardware implementations.
In contrast, 
 \cite{contracts2021,mosier2022axiomatic} propose frameworks for formalizing security contracts between hardware and software.
Our notion of contract satisfaction (\Cref{def:contract-satisfaction}) is inspired by the framework from~\cite{contracts2021}, which we instantiate and adapt for reasoning about RTL processors.

\section{Conclusion}\label{sec:conclusions}

We presented an approach for verifying RTL processor designs against ISA-level leakage contracts.
We implemented our approach in the \tool{} verification tool, which we use to  characterize the side-channel security guarantees of three open-source RISC-V processors.
This demonstrates that leakage contracts can be successfully applied to RTL processor designs.
It also paves the way for linking recent advances on specification~\cite{contracts2021,mosier2022axiomatic} and software analysis~\cite{spectector2020,fabian202automatic,pitchfork, GuancialeBD20, blade} for leakage contracts to RTL processor designs. %

\begin{acks}
We would like to thank Alastair Reid and Piotr Sapiecha for feedback and discussions.
This project has received funding from the \grantsponsor{1}{European Research Council}{https://erc.europa.eu/} under the European Union's Horizon 2020 research and innovation programme (grant agreement No. \grantnum{1}{101020415}), 
from the \grantsponsor{2}{Spanish Ministry of Science and Innovation}{https://www.ciencia.gob.es/} under the project \grantnum{2}{TED2021-132464B-I00 PRODIGY}, 
from the \grantsponsor{3}{Spanish Ministry of Science and Innovation}{https://www.ciencia.gob.es/} under the Ram\'on y Cajal grant \grantnum{3}{RYC2021-032614-I}, 
from the \grantsponsor{3}{Spanish Ministry of Science and Innovation}{https://www.ciencia.gob.es/} under the project \grantnum{3}{PID2022-142290OB-I00 ESPADA}, 
and from a gift by Intel Corporation.\looseness=-1
\end{acks}

\bibliographystyle{IEEEtran}
\balance
\bibliography{references}

\onlyTechReport{
\clearpage
\appendix

\section{$\lang{}$ semantics}\label[appendix]{app:semantics}

The full semantics of $\lang{}$ is given in \Cref{figure:language:semantics}.
In the figure, $\emptysequence$ denotes the empty sequence and $\concat$ denotes the concatenation operator. 
As is standard, we have that $\tau \concat \emptysequence = \tau$ and $\emptysequence \concat \tau = \tau$.
As mentioned in \Cref{sec:model:language}, the trace semantics extends the notion of valuation to refer to both registers and wire variables in a circuit's output as shown in the definition of the projection $\project{\mu}{C}(i)$.

\begin{figure}
    \begin{align*}
        \hwEval{C}{\mu} &:= \{ (r,\hwEvalC{e}{\mu}{C}) \in \Var \times \Val \mid \passign{r}{e} \in C.A \} \\
        & \quad \cup \{ (r,\mu(r)) \in \Var \times \Val \mid r \not\in \writeVars{C} \}\\
        \hwEvalC{n}{\mu}{C} &:= n \\
        \hwEvalC{r}{\mu}{C} &:= \mu(r) \\
        \hwEvalC{v}{\mu}{C} &:= \hwEvalC{e}{\mu}{C}\ \text{where}\ v = e \in \wires{C} \\
        \hwEvalC{\unaryOp{e}}{\mu}{C} &:= \unaryOp{(\hwEvalC{e}{\mu}{C}) }\\
        \hwEvalC{\binaryOp{e_1}{e_2}}{\mu}{C} &:= \binaryOp{(\hwEvalC{e_1}{\mu}{C})}{(\hwEvalC{e_2}{\mu}{C})}\\
        \hwEvalC{\ite{e_1}{e_2}{e_3}}{\mu}{C} &:= 
            {
                \begin{cases}
                    \hwEvalC{e_2}{\mu}{C} & \text{if}\ \hwEvalC{e_1}{\mu}{C} \neq 0\\
                    \hwEvalC{e_3}{\mu}{C} & \text{if}\ \hwEvalC{e_1}{\mu}{C} = 0
                \end{cases}
            }\\
        \hwEvalC{e_1[e_2:e_3]}{\mu}{C} &:= (\hwEvalC{e_1}{\mu}{C})[ (\hwEvalC{e_2}{\mu}{C}) : (\hwEvalC{e_3}{\mu}{C})]\\[3mm]
        \hwEval{C}{\mu, 0} &:= \mu\\
        \hwEval{C}{\mu, n+1} &:= \hwEval{C}{\hwEval{C}{\mu, n}} \\[3mm]
        \hwEvalInfty{C}{\mu} &:= \mu_0 \concat \mu_1 \concat  \mu_2  \concat \dots \\
        & \text{where}\ \mu_i = \project{\hwEval{C}{\mu, i}}{C} \\[3mm]
        \hwEvalInftyFilter{C}{\mu}{\phi} &:= \mu_0 \concat \mu_1 \concat \mu_2 \concat \dots\\
                & \text{where}\ \mu_i =   \project{\hwEval{C}{\mu, i}}{C}\ \text{if}\ 
                C, \hwEval{C}{\mu, i} \models \phi \\
                & \text{and}\ \mu_i = \emptysequence\ \text{otherwise}\\[3mm]
        C, \val \models \phi &\text{ iff } \hwEvalC{\phi}{\val}{C} \neq 0 \\[3mm]
        \project{\mu}{V}(r) &:= {
            \begin{cases}
                \mu(r) & \text{if}\ i \in V \\  
                \bot & \text{otherwise}
            \end{cases}
        }      
        \\[3mm]
        \project{\mu}{C}(i) &:= {
                \begin{cases}
                    \mu(i) & \text{if}\ i \in \outputs{C} \cap \Var \\
                    \hwEvalC{i}{\mu}{C} & \text{if}\ i \in \outputs{C} \cap \WireVars \\
                    \bot   & \text{otherwise}
                \end{cases}
        } 
        \\[3mm]
        \val \equivVars{V} \val' & \text{ iff } \qquad  \project{\val}{V} = \project{\val'}{V}
    \end{align*}
    \vspace{-15pt}
    \caption{$\lang$ semantics. 
    }\label{figure:language:semantics}
    \end{figure}
\section{A temporal logic for $\lang{}$}\label[appendix]{app:temporal-logic}

Here, we introduce a logic for expressing temporal properties of $\lang$ circuits.
Formulas $\Phi$ in this logic are constructed by combining $\lang$ predicates $\phi$ with temporal operators $\future$ (denoting ``in the next cycle''), 
$\boundedFuture{B}$ (denoting ``for the next $B$ cycles''), 
$\alwaysFuture$ (denoting ``always in the future''),  
and the usual boolean operators. %
Its semantics is the following: 
\begin{align*}
    C & \models \Phi  \text{ if } C,\mu,0 \models \Phi \text{ for all } \mu \in \states{C}\\
    C,\mu,i & \models \phi  \text{ if } C, \hwEval{C}{\mu, i} \models \phi \\
    C,\mu,i & \models \future \Phi  \text{ if } C,\mu,i+1 \models \Phi \\
    C,\mu,i & \models \boundedFuture{k} \Phi \text{ if } C,\mu,i+j \models \Phi \text{ for all } j < k\\ 
    C,\mu,i & \models \alwaysFuture \Phi  \text{ if } C,\mu,i+j \models \Phi \text{ for all } j \in \Nat\\
\end{align*}

\section{Proofs}\label[appendix]{app:proofs}

Here, we present the proofs of \Cref{thm:contract-sat-preconditions} and \Cref{thm:soundness}, which we restate below for simplicity.

\begin{theorem*}[Decoupling Theorem]
If $\isasat{\uarch}{\arch}{\phi}$ holds for retirement predicate $\phi$, then
\[ \leakorder{ \contract }{ \atk }{\uarch}{\uarchVars, \phi} \Leftrightarrow \ctrsat{ \compose{\contract}{\arch}}{ \compose{\atk}{\uarch}}. \]
    \end{theorem*}
    
    \begin{proof}
    We assume $\isasat{\uarch}{\arch}{\phi}$ and prove the two directions.
    
    \paragraph{$ \boldsymbol{\Rightarrow}$}
    Let $\val,\val' \in \initStates{\uarch}$ be  such that $\val \equivVars{\uarchVars} \val'$ and $$\hwEvalInfty{ \compose{\contract}{\arch} }{ \val } = \hwEvalInfty{ \compose{\contract}{\arch} }{ \val' }.$$
    Since $\isasat{\uarch}{\arch}{\phi}$ and $\contract$ is a monitoring circuit for $\arch$, we get $$ \hwEvalInftyFilter{ \compose{\contract}{\uarch} }{ \val }{ \phi } = \hwEvalInftyFilter{ \compose{\contract}{\uarch} }{\val'}{\phi}. $$
    From $ \leakorder{ \contract }{ \atk }{\uarch}{\uarchVars, \phi}$, we get  $$ \hwEvalInfty{ \compose{\atk}{\uarch} }{ \val }= \hwEvalInfty{ \compose{\atk}{\uarch} }{\val'}. $$
    Therefore, $\ctrsat{ \compose{\contract}{\arch}}{ \compose{\atk}{\uarch} }$.
    
    \paragraph{$ \boldsymbol{\Leftarrow}$}
    Let $\val,\val' \in \initStates{\uarch}$ be such that $\val \equivVars{\uarchVars} \val'$ and $$ \hwEvalInftyFilter{ \compose{\contract}{\uarch} }{ \val }{ \phi } = \hwEvalInftyFilter{ \compose{\contract}{\uarch} }{\val'}{\phi}. $$
    Since $\isasat{\uarch}{\arch}{\phi}$ and $\contract$ is a monitoring circuit for $\arch$, we get $$\hwEvalInfty{ \compose{\contract}{\arch} }{ \val } = \hwEvalInfty{ \compose{\contract}{\arch} }{ \val' }.$$
    From $ \ctrsat{ \compose{\contract}{\arch}}{ \compose{\atk}{\uarch}}$, we get  $$ \hwEvalInfty{ \compose{\atk}{\uarch} }{ \val }= \hwEvalInfty{ \compose{\atk}{\uarch} }{\val'}. $$
    Therefore, $\leakorder{ \contract }{ \atk }{\uarch}{\uarchVars, \phi}$.
    \end{proof}

\begin{theorem*}
    $\textsc{Verify}(\uarch$, $\contract$, $\atk$, $\phi$,  $ b $, $ \RelInvs) \Rightarrow \leakorder{ \contract }{ \atk }{\uarch}{\uarchVars, \phi}$
\end{theorem*}

\renewcommand{\CandInvs}{\mathit{CI}}

\begin{proof}
We split the proof in two steps.

\paragraph{Soundness of \textsc{LearnInv}}
Here, we show that the outcome of invariant learning, $LI := \textsc{LearnInv}(C, \Phi_{\mathit{initial}}, \Phi_{assumption}, b, \CandInvs)$ is a set of  invariants of $C$ in all executions that satisfy $\Phi_{assumption}$ in every cycle.
In other words, $C \models (\Phi_{\mathit{initial}} \wedge \alwaysFuture \Phi_{assumption}) \to \alwaysFuture \bigwedge LI$.

Let $\val$ be an arbitrary valuation for $C$ such that (a) $C, \val, 0 \models \Phi_{\mathit{initial}}$ and (b) $C, \val, 0 \models  \alwaysFuture \Phi_{assumption}$.
We now show, by induction on $i$, that $C, \val, i \models \bigwedge LI$.
\begin{description}
\item[Base case:] 
We need to show that $C, \val, 0 \models \bigwedge LI$ holds.
Since $LI$ has been returned by \textsc{LearnInv}, we know that $C \models \Psi_{\mathit{base}}$ holds for a set of invariants $CI \supseteq LI$.
From this, we have that $C, \val, 0 \models (\Phi_{\mathit{initial}} \wedge \boundedFuture{b} \Phi_{\mathit{assumption}}) \to \bigwedge LI$.
From (a) and (b), we get $C, \val, 0 \models \Phi_{\mathit{initial}} \wedge \boundedFuture{b} \Phi_{\mathit{assumption}}$.
Thus we can conclude $C, \val, 0 \models \bigwedge LI$.

\item[Induction step:]
We now show that $C, \val, i \models \bigwedge LI$ holds given that $C, \val, j \models \bigwedge LI$ holds for all $j < i$.
Let $\val'$ be the valuation reached in $i-1$ steps from $\val$.
Since $LI$ has been returned by \textsc{LearnInv}, we know that $C \models \Psi_{\mathit{inductive}}$ holds for $\CandInvs = LI$.
From this, we have that $C, \val', 0 \models (\bigwedge LI \wedge \boundedFuture{b} \Phi_{\mathit{assumption}})  \to \future \bigwedge LI$ holds.
From the induction hypothesis, we have that $C, \val, i-1 \models \bigwedge LI$ holds and, therefore, we get $C, \val', 0 \models \bigwedge LI$.
From $C, \val, 0 \models \alwaysFuture \Phi_{assumption}$, we also get $C, \val', 0 \models \boundedFuture{b} \Phi_{\mathit{assumption}}$.
Therefore, we can derive $C, \val', 0 \models \future \bigwedge LI$.
From this, we get $C, \val', 1 \models \bigwedge LI$.
From this and $\val'$ being reached from $\val$ in $i-1$ steps, we get $C, \val, i \models \bigwedge LI$.
\end{description}

\paragraph{Soundness of \textsc{Verify}}
Assume, for contradiction's sake, that $\textsc{Verify}(\uarch$, $\contract$, $\atk$, $\phi$, $ \archVars $, $ b$, $ \RelInvs) = \top$ and  $\leakorder{ \contract }{ \atk }{\uarch}{\uarchVars, \phi}$ does not hold.
From the latter, there are two $\uarchVars$-equivalent initial valuations $\val,\val'$ such that $\hwEvalInftyFilter{\compose{\contract}{\uarch}}{\mu}{\phi}$ $=$ $\hwEvalInftyFilter{\compose{\contract}{\uarch}}{\mu'}{\phi}$ and $\hwEvalInfty{\compose{\atk}{\uarch}}{ \mu }$ $\neq$ $\hwEvalInfty{\compose{\atk}{\uarch}}{ \mu'}$.
Thus:
\begin{enumerate}
\item From $\mu,\mu'$ being initial valuations, we have $\stutteringProduct{\uarch}{\phi}, \mu \times \mu', 0 \models \runA{\psi_{init}^{\uarch}} \wedge \runB{\psi_{init}^{\uarch}}$, where $\mu \times \mu'$ is the valuation defined as $\mu \times \mu'(\runA{v}) = \mu(v)$ and $\mu \times \mu'(\runB{v}) = \mu'(v)$ for all $v \in \vars{\uarch}$.
\item From $\mu,\mu'$ being $\uarchVars$-equivalent, we have $\stutteringProduct{\uarch}{\phi}, \mu \times \mu', 0 \models \psi_{\mathit{equiv}}^{\uarchVars}$.
\item From $\hwEvalInftyFilter{\compose{\contract}{\uarch}}{\mu}{\phi} = \hwEvalInftyFilter{\compose{\contract}{\uarch}}{\mu'}{\phi}$, we have that $\stutteringProduct{\uarch}{\phi}, \mu \times \mu', 0 \models \alwaysFuture \Phi_{\mathit{ctr}}$.
\item Finally, from $\hwEvalInfty{\compose{\atk}{\uarch}}{ \mu }$ $\neq$ $\hwEvalInfty{\compose{\atk}{\uarch}}{ \mu'}$, we have that  $\product{\uarch}{\uarch}, \mu \times \mu', 0 \not\models \alwaysFuture \psi_{\mathit{equiv}}^\atk$.
Moreover, from $\textsc{Verify}(\uarch$, $\contract$, $\atk$, $\phi$, $ b$, $ \RelInvs) = \top$, we have that  $\stutteringProduct{\uarch}{\phi} \models \alwaysFuture \runA{\psi} \leftrightarrow \runB{\psi}$.
Therefore, from $\product{\uarch}{\uarch}, \mu \times \mu', 0 \not\models \alwaysFuture \psi_{\mathit{equiv}}^\atk$, we also get $\stutteringProduct{\uarch}{\phi}, \mu \times \mu', 0 \not\models \alwaysFuture \psi_{\mathit{equiv}}^\atk$ because the stuttering never happens since $\phi$ is always synchronized in the two executions. 
\end{enumerate}
Moreover, from $\textsc{Verify}(\uarch$, $\contract$, $\atk$, $\phi$, $ b$, $ \RelInvs) = \top$, there is a set $\CandInvs$ returned by \textsc{LearnInv} such that $\stutteringProduct{\uarch}{\phi} \models \bigwedge_{\psi \in \CandInvs} \psi \to \psi^{\atk}_{\mathit{equiv}}$.
From (1)--(3) and the soundness of \textsc{LearnInv} (proved above), we have that $\stutteringProduct{\uarch}{\phi}, \mu \times \mu', 0 \models \alwaysFuture \bigwedge_{\psi \in \CandInvs} \psi$.
From this and $\stutteringProduct{\uarch}{\phi} \models \bigwedge_{\psi \in \CandInvs} \psi \to \psi^{\atk}_{\mathit{equiv}}$, we have $\stutteringProduct{\uarch}{\phi}, \mu \times \mu', 0 \models \alwaysFuture \psi^{\atk}_{\mathit{equiv}}$, which contradicts (4).\looseness=-1
\end{proof}

\section{\FourWayTool{} verification approach}

The approach from \FourWayTool{} (used in \Cref{sec:evaluation:reduction}) is given in \Cref{alg:verification-4way}.

\newcommand{\runX}[2]{{#1}^{#2}}

\begin{algorithm}[t]
    \caption{\FourWayTool{} verification approach}\label{alg:verification-4way}
\begin{algorithmic}[1]
	\Require Microarchitecture $\uarch$, architecture $\arch$, leakage monitor $\contract$, attacker $\atk$, retirement predicate $\phi$, lookahead $b$, candidate invariants $CI$\looseness=-1
		\Statex{}

    \Procedure{Verify}{$\uarch , \contract, \atk, \phi, b, CI$}\label{line:alg-verification-4way:begin-verify}
        \State{$\psi_{\mathit{equiv}}^{\uarchVars} := \bigwedge_{r \in \uarchVars} \runX{r}{3} = \runX{r}{4}$}
        \State{$\psi_{\mathit{archEquiv}} := \bigwedge_{r \in \archVars} \runX{r}{1} = \runX{r}{3} \wedge \runX{r}{2} = \runX{r}{4}$}
        \State{$\Phi_{\mathit{initial}} := \runX{\psi_{\mathit{init}}^{\arch}}{1} \wedge \runX{\psi_{\mathit{init}}^{\arch}}{2} \wedge  \runX{\psi_{\mathit{init}}^{\uarch}}{3} \wedge \runX{\psi_{\mathit{init}}^{\uarch}}{4} \wedge \psi_{\mathit{archEquiv}} \wedge \psi_{\mathit{equiv}}^{\uarchVars} $}\label{line:alg-verification-4way:initial-formula}
        \State{$\Phi_{\mathit{ctr-equiv}} := \bigwedge_{o \in \outputs{\contract}} \runX{o}{1} = \runX{o}{2}$}\label{line:alg-verification-4way:precondition-formula}
        \State{$LI := \textsc{LearnInv}(\product{\arch}{\product{\arch}{\product{\uarch}{\uarch}}} , \Phi_{\mathit{initial}}, \Phi_{\mathit{ctr-equiv}}, b, CI)$}\label{line:alg-verification-4way:learninv-call}
        \State{\Return{$\product{\arch}{\product{\arch}{\product{\uarch}{\uarch}}}\models \bigwedge LI \to \bigwedge_{o \in \outputs{\atk}} \runX{o}{3} = \runX{o}{4} $}}\label{line:alg-verification-4way:security-check} 
    \EndProcedure\label{line:alg-verification-4way:end-verify}

\end{algorithmic}
\end{algorithm}
}
\end{document}